\documentclass[12pt,reqno]{amsart}

\usepackage{hyperref}
\hypersetup{
colorlinks=true,
linkcolor={blue}
}

\usepackage[font=small,skip=0pt]{caption}
\usepackage{amsmath}
\usepackage{amsthm}
\usepackage{amssymb}
\usepackage{graphicx}
\usepackage{epstopdf}
\graphicspath{{Figures/}}

\newtheorem{theorem}{Theorem}
\newtheorem{definition}{Definition}

\newtheorem{remark}{Remark}
\newtheorem{proposition}{Proposition}

\newcommand{\R}{\mathbb{R}}

\newcommand{\sech}{\text{sech}}

\newcommand{\mL}{\mathcal L}

\newcommand{\lm}{\lambda}
\newcommand{\eps}{\varepsilon}

\begin{document}

\title[Optical systems with an intensity-dependent dispersion]{\bf Localization in optical systems with an intensity-dependent dispersion}

       \author{R.M. Ross, P.G. Kevrekidis}
       \address{Department of Mathematics and Statistics,
         University of Massachusetts, Amherst, MA 01003-4515, USA}

       \author{D.E. Pelinovsky}
\address{Department of Mathematics and Statistics, McMaster
  University, Hamilton, Ontario, Canada, L8S 4K1}

\date{\today}

\begin{abstract}
We address the nonlinear Schr{\"o}dinger equation with  intensity-dependent dispersion which was recently proposed in the context of nonlinear optical systems. Contrary to the previous findings, we prove that no solitary wave 
solutions exist if the sign of the intensity-dependent dispersion 
coincides with the sign of the constant dispersion, whereas 
a continuous family of such solutions exists in the case of the opposite signs. The family includes two particular solutions, 
namely cusped and bell-shaped solitons, where the former 
represents the lowest energy state in the family and the latter 
is a limit of solitary waves in a regularized 
system. We further analyze the delicate analytical properties 
of these solitary waves such as the asymptotic behavior near singularities, the
spectral stability, and the convergence of the fixed-point iterations 
near such solutions. The analytical theory is corroborated 
by means of numerical approximations. 
\end{abstract}

\maketitle

\section{Introduction}

The study of solitary waves in nonlinear Schr{\"o}dinger (NLS) type
equations~\cite{ablowitz,ablowitz2,ablowitz1,sulem}
is a topic of wide interest in a broad range of disciplines. This is
because of the ubiquitous nature of the relevant envelope wave equation 
which appears in settings as
diverse as the propagation of the electric field in optical fibers \cite{hasegawa,agrawal}, the
evolution of the probability density of atoms in Bose-Einstein
condensates~\cite{pethick,stringari}, but also in nonlinear waves in
plasmas~\cite{zakh2},
and freak waves in the ocean~\cite{slunaev}. In the simplest case
of bright~\cite{sulem,agrawal} and dark~\cite{siambook} solitary waves
the interplay of linear, constant coefficient dispersion and cubic nonlinearity (e.g., stemming from the Kerr effect in optics~\cite{hasegawa,agrawal} or
a mean-field approximation in Bose-Einstein condensation~\cite{pethick,stringari}) leads to the Duffing differential equation for the spatial profile of the solitary
wave. The spatial profile is smooth and decays exponentially to 
zero or to a  nonzero constant background.

In recent years, however, there has been an increasing interest in the
study of systems that feature intensity-dependent dispersion (IDD).
There exist multiple relevant examples of such systems, 
ranging from femtosecond pulse propagation in quantum well
waveguides~\cite{koser} to electromagnetically induced transparency in
coherently prepared multistate
atoms~\cite{greentree}. A recent work on this subject in ~\cite{OL2020} introduced a prototypical example of IDD and addressed non-standard types 
of solitary wave solutions of the NLS equation with IDD. Two different signs
of the intensity dependence were considered: one being the same as that of
linear dispersion and the other being opposite to that of linear dispersion.

{\em The purpose of this work} is to follow the intriguing example
of the NLS equation with IDD 
and to examine the relevant solitary wave solutions 
in detail. Contrary to the previous 
findings in~\cite{OL2020}, we prove that one of the two solutions 
examined earlier, namely {\em the cusped soliton}, does not exist 
in the case of the same sign of IDD but exists in the case of the 
opposite sign of IDD. In the latter case, it is a member of the 
continuous family of solitary wave solutions, which includes 
{\em the bell-shaped soliton} explored in \cite{OL2020}.

Periodic in space solutions are also possible in the model with the opposite sign of IDD. We briefly mention these periodic solutions but 
focus mainly on the existence and stability of the solitary wave solutions 
in the NLS equation with IDD.

\subsection{Main results}

We address the following NLS equation with IDD:
\begin{align}
\label{nls}
i\psi_t + (1-b|\psi|^2)\psi_{xx} = 0,
\end{align}
where $\psi = \psi(x,t)$ is the complex wave function and
$b$ is a real parameter. It was shown in \cite{OL2020} that
the NLS equation (\ref{nls}) admits formally two conserved quantities:
\begin{equation}
\label{mass-energy}
Q(\psi) = -\frac{1}{b} \int_{\mathbb{R}} \log|1-b|\psi|^2| dx, \quad
E(\psi) = \int_{\mathbb{R}} |\psi_x|^2 dx.
\end{equation}
The two conserved quantities have the meaning of the mass and energy
of the optical system and they are related to the phase rotation 
($\psi \mapsto \psi e^{i \theta}$, $\theta \in \mathbb{R}$)
and the time translation ($\psi(x,t) \mapsto \psi(x,t+t_0)$, $t_0 \in \mathbb{R}$) symmetries of the NLS equation (\ref{nls}).
The conserved quantities (\ref{mass-energy}) 
are defined in the subspace of $H^1$ functions given by
\begin{equation}
\label{energy-space}
X = \left\{ u \in H^1(\mathbb{R}) : \quad \left| \int_{\mathbb{R}} \log|1-b|\psi|^2| dx \right| < \infty \right\},
\end{equation}
which is the energy space of the NLS equation (\ref{nls}).

The standing wave solutions are given by
\begin{align}
\label{standing-wave}
\psi(x,t) = e^{ict} u(x)
\end{align}
where $c$ is a real parameter and $u(x)$ satisfies the differential equation
\begin{align}
\label{ode}
c u = (1-bu^2) u''(x).
\end{align}
Since the linear Schr\"{o}dinger equation $i \psi_t + \psi_{xx} = 0$ 
admits the linear waves $\psi(x,t) \sim e^{ikx - ik^2 t}$ 
which corresponds to $c = - k^2 \leq 0$, 
the true localization is possible only if $c > 0$, for which 
tails of solitary waves avoid resonance with the linear waves.

Let us now give the definition of the weak solutions of the differential equation (\ref{ode}). 

\begin{definition}
	\label{def-weak-solution}
	We say that $u \in H^1(\mathbb{R})$ is a weak solution 
	of the differential equation (\ref{ode}) if it satisfies 
	the following equation 
	\begin{equation}
	\label{ode-weak}
	c \langle u, \varphi \rangle + \langle (1- b u^2) u', \varphi' \rangle 
	- 2 b \langle u (u')^2, \varphi \rangle = 0, \quad \mbox{\rm for every } \varphi \in H^1(\mathbb{R}),
	\end{equation}
	where $\langle \cdot, \cdot \rangle$ is the standard inner product in $L^2(\mathbb{R})$. We say that the solution is positive if $u(x) > 0$ for every $x \in \mathbb{R}$ and single-humped if there exists only one point $x_0 \in \mathbb{R}$ such that $u(x_0) = \max\limits_{x \in \mathbb{R}} u(x)$.
\end{definition}

We study the weak solutions in Definition \ref{def-weak-solution} by looking 
for the smooth orbits of the second-order differential equation 
(\ref{ode}), see Propositions \ref{prop-soliton-1} and \ref{prop-soliton-2}. The orbits remain smooth if $b < 0$ but 
have singularities if $b > 0$. With the precise analysis of the asymptotic behavior of the solutions near the singularities (similar to the analysis 
in the recent work \cite{Alfimov}), we prove that the solutions 
remain in $H^1(\mathbb{R})$ across the singularity points.

The following theorem formulates the first main result of the paper.

\begin{theorem}
	\label{theorem-main1}
	Fix $c > 0$ and consider weak, positive, and single-humped solutions of Definition \ref{def-weak-solution}.
	No such solutions exist for $b < 0$, 
	whereas a one-parameter continuous family of such solutions 
	exists for each $b > 0$ in the energy space $X$.
\end{theorem}

Two particular solitary wave solutions of the 
continuous family in Theorem \ref{theorem-main1} for $b = 1$ and $c = 1$ 
are shown on Fig. \ref{bluesolns}. 
We call them the {\em cusped} and {\em bell-shaped} solitons as 
shown on the left and right panels, respectively.
Without loss of generality, the solutions can be translated to be even in $x$.
The cusped soliton satisfies $0 < u(x) \leq 1$ with the only singularity 
at $u(0) = 1$. The bell-shaped soliton satisfies $0 < u(x) \leq \sqrt{2}$ 
with two singularities at $u(\pm \ell) = 1$ for a uniquely defined
$\ell > 0$. The singular behavior of these solutions is 
further clarified in Proposition \ref{prop-cusped}.

\begin{figure}[hbt]
	\centering
	\includegraphics[width=6cm,height = 4cm]{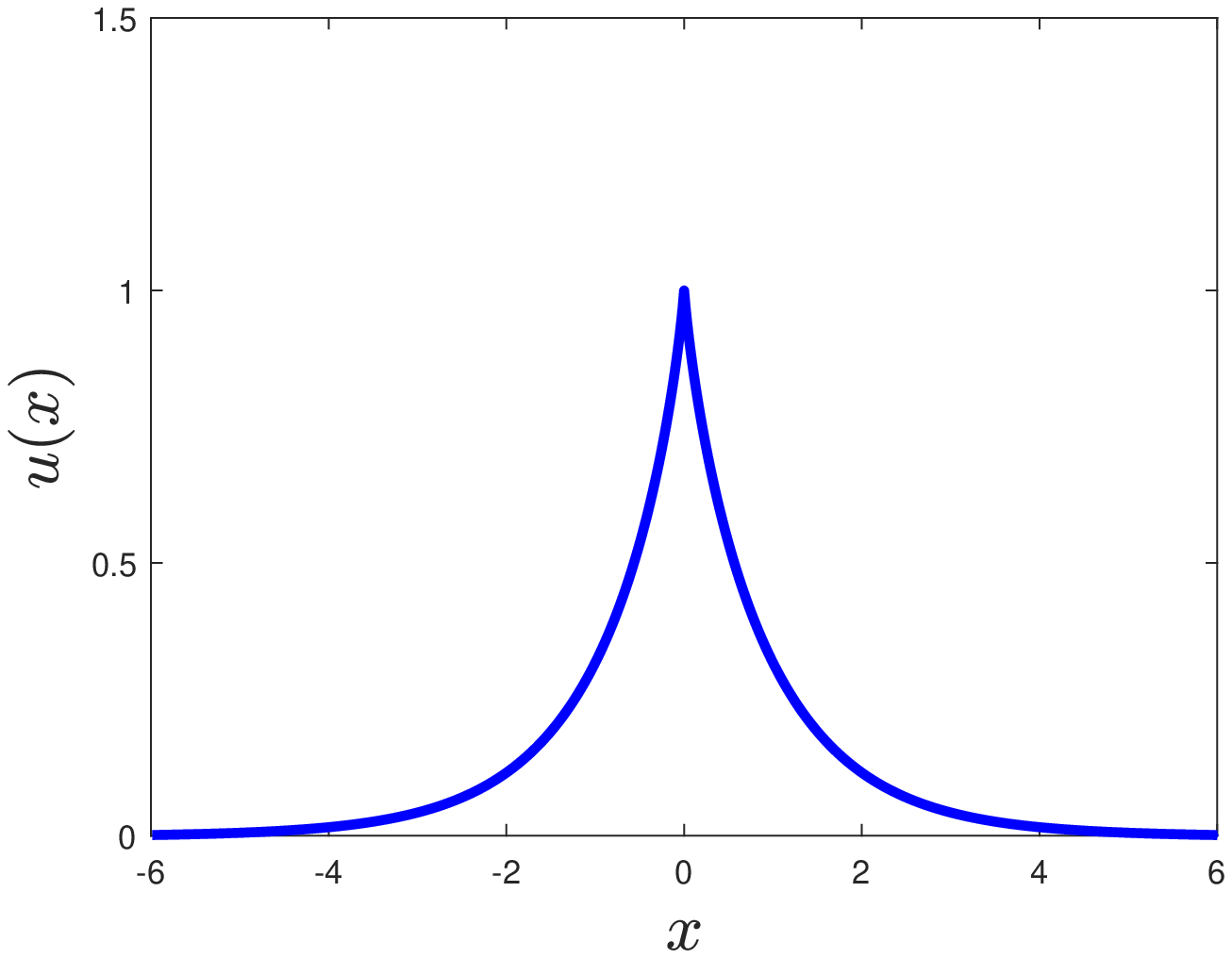}
	\includegraphics[width=6cm,height = 4cm]{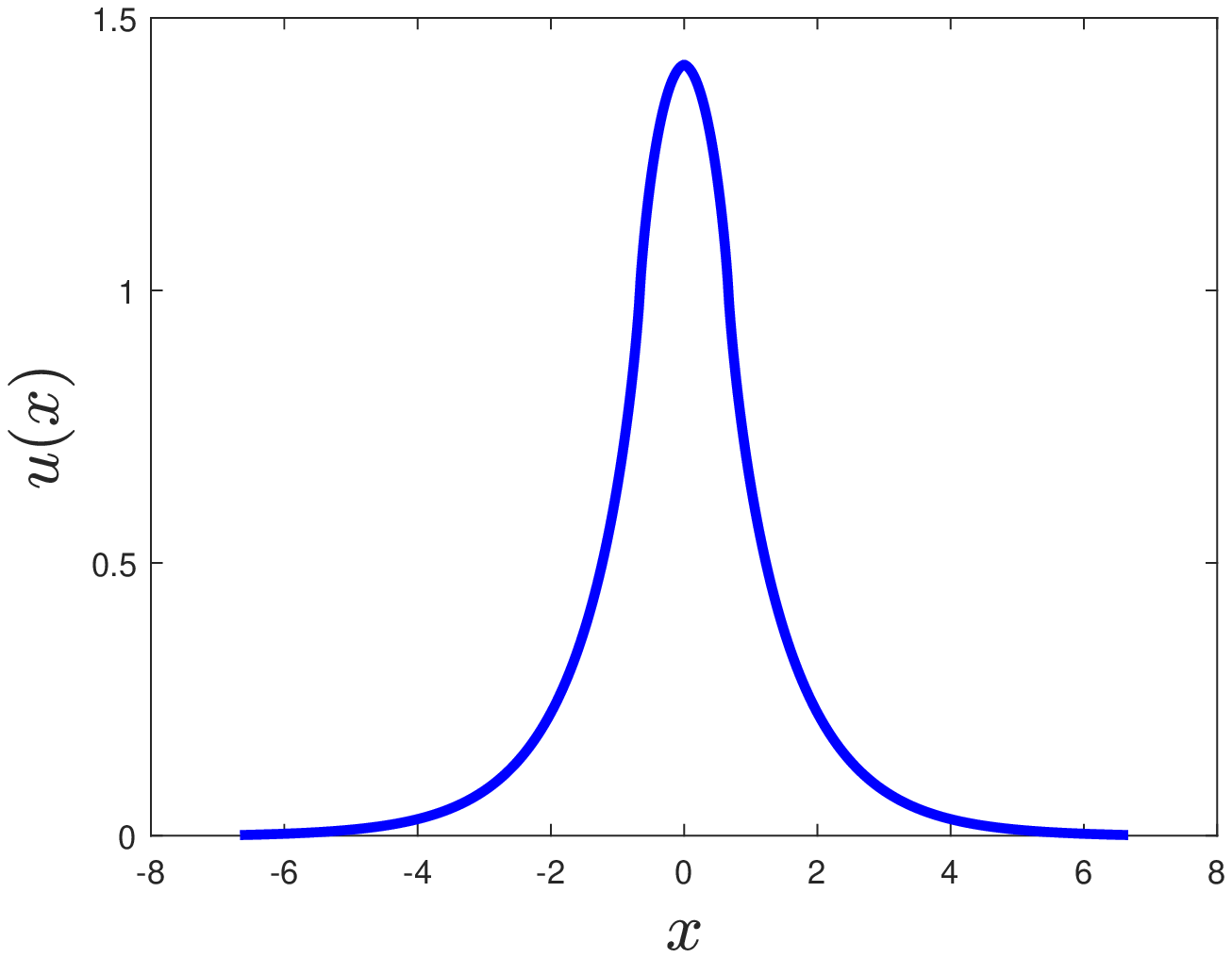} 
	\caption{The spatial profiles $u(x)$ of the two single-humped solitary wave solutions of the second-order equation (\ref{ode}) for $c = 1$ and $b = 1$: cusped soliton (left) and bell-shaped soliton (right).}
	\label{bluesolns}
\end{figure}

\begin{remark}
The result of Theorem \ref{theorem-main1} disagrees with the numerical results in \cite{OL2020}, where the solitary wave solutions were also obtained 
for $b < 0$. According to Theorem \ref{theorem-main1}, such solutions do not exist. In the case $b > 0$, the bell-shaped soliton was obtained in \cite{OL2020}, however, the cusped soliton and the continuous family of solitary wave solutions were missed  in \cite{OL2020}.
\end{remark}

The second main result of this paper is about 
numerical approximations of the cusped and bell-shaped 
solitons. We implement three numerical methods towards identifying
these waves and 
elaborate on convergence of these methods in 
the neighborhood of the cusped and bell-shaped solitons in $H^1(\mathbb{R})$. 
The outcomes of this study are summarized as follows:\\

\begin{itemize}
	\item Regularization of the differential equation (\ref{ode}) for $b = c = 1$ near the singularities $u = \pm 1$ allows us to approximate the bell-shaped 
	soliton only. We prove in Proposition \ref{prop-regularization} that the sequence of regularized solitary wave solutions converges in $H^1(\mathbb{R})$ to the bell-shaped soliton. \\
	
	\item Fixed-point iterations with the popular Petviashvili's
          method ~\cite{petviashvili} (also referred to
          as the spectral renormalization method~\cite{amus})
          allows us to approximate the cusped soliton only. We prove in Propositions \ref{prop-Petv} and \ref{prop-Petv-cusped} that the method diverges for the bell-shaped soliton and for other solitary wave solutions. The cusped soliton represents the lowest energy state 
	in the continuous family of solitary waves.\\
	
	\item Fixed-point iterations with the regular Newton's method 
	allow us to approximate both the bell-shaped and cusped
	solitons, as well as arbitrary members within the continuous family of solitary
	waves upon suitable initial guesses. We are able to prove convergence 
	of the Newton's method near the cusped soliton in Proposition \ref{prop-Newton-cusped}.\\
\end{itemize}

The third main result of this paper is about stability of solitary waves 
with respect to small perturbations in the time evolution of the NLS
equation (\ref{nls}). Due to singularities of the solitary wave
solutions, we conclude that the
analysis of stability is an open mathematical problem even at the 
level of {\em spectral stability}. We are only able to characterize 
the kernel of the linearized operator and only in the case of the cusped 
soliton in Proposition \ref{prop-kernel}. Nevertheless, numerical 
approximations of eigenvalues of the discretized and truncated 
spectral stability problem suggest that {\em cusped and bell-shaped solitons 
are spectrally stable.}

The same conclusion regarding the dynamical
stability of the cusped and bell-shaped solitons 
is supported by the results of direct numerical simulations of the NLS equation (\ref{nls}). For time integration, we use a pseudospectral method with the Fourier transform in the spatial domain $[-30,30]$ with $N=2048$ points. In order to solve the time-evolution equations for Fourier modes, we use the fourth-order Runge-Kutta method with time step $\Delta t = 0.001$. 

\begin{figure}[hbt]
	\centering
	\includegraphics[width=6cm,height=4cm]{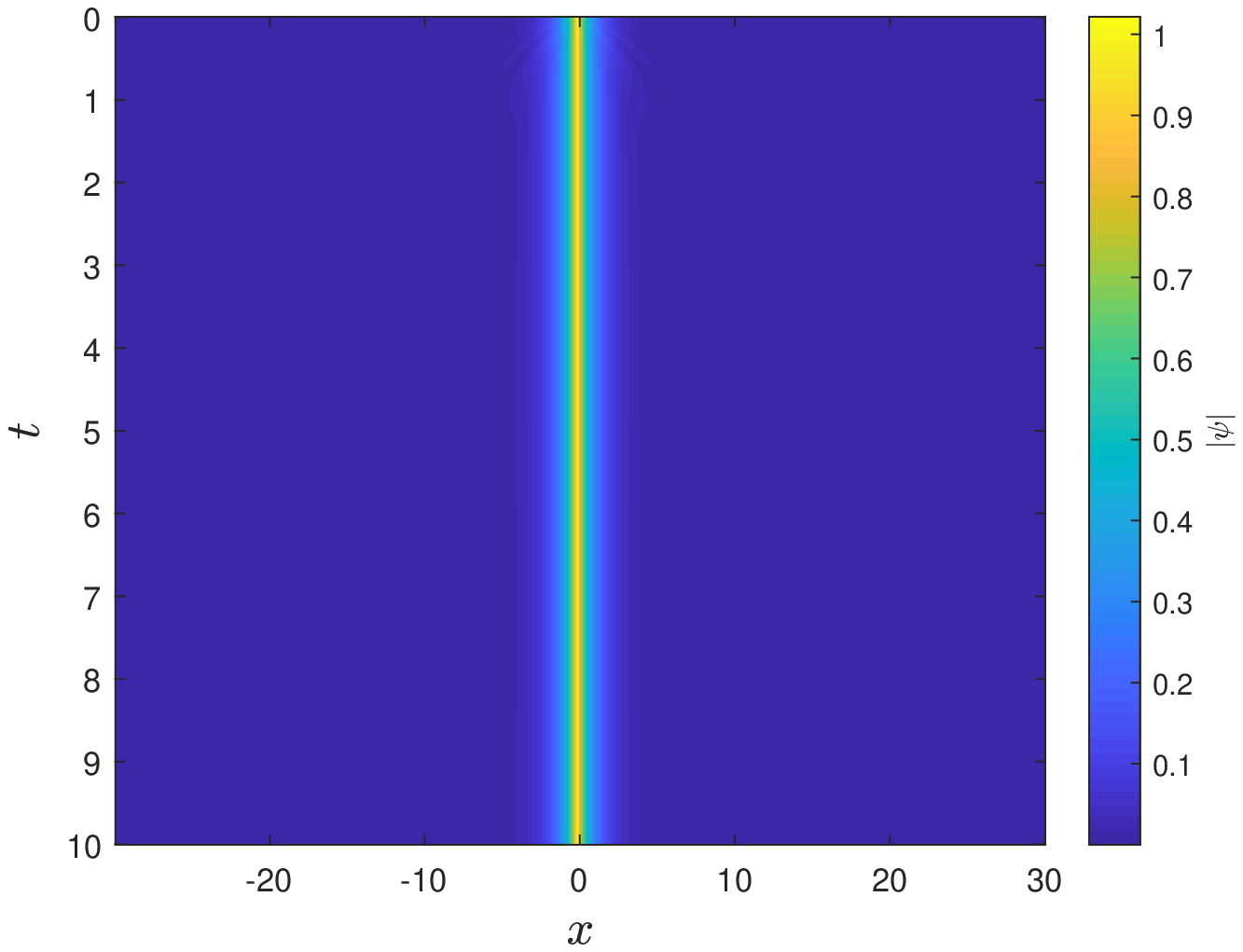}
	\includegraphics[width=6cm,height=4cm]{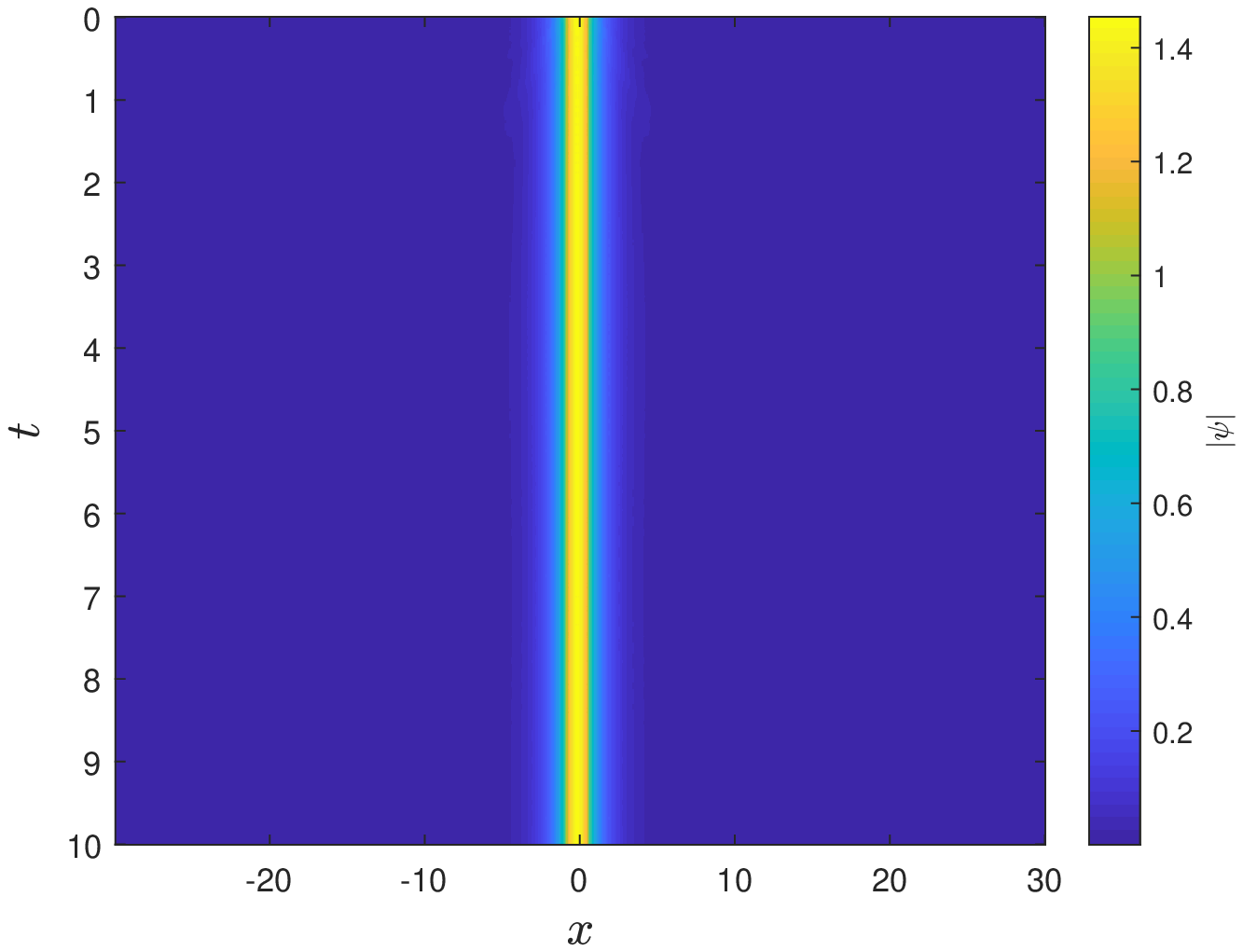}
	\caption{The space-time evolution of the NLS equation
		(\ref{nls}) with the initial perturbations 
		of the cusped soliton (left) and bell-shaped soliton
		(right). In this and similar space-time figures, 
		the contour plot is of the wavefunction modulus $|\psi|$.}
	\label{bluesolns-time}
\end{figure}

Figure \ref{bluesolns-time} presents outcomes of the numerical simulations 
of the initial conditions taken as perturbations of the solitary wave solutions 
$\psi(x,0) = 1.01u(x)$. 
The evolution of these waveforms is (nearly) steady and the small
perturbations
disperse away from the stationary localized solution.
Notice that, in the vicinity of the boundary, a dissipative layer has
been
used, absorbing the small amplitude wavepackets originally emitted
by the localized waves. Simulations for
considerably
longer times have also been performed and we have confirmed stability of 
both solitons in longer computations and under different types of
small
perturbations.

The methods and results obtained in the analytical and numerical parts of this work are very similar to the recent study of compactons 
in the degenerate NLS equation in \cite{GHM} and in the sublinear KdV equation 
in \cite{PelinPelin}.

\subsection{Organization of the paper}

Our presentation is structured as follows. 

In Section 2, we study the smooth orbits of the differential equation (\ref{ode}). The asymptotic behavior of the solitary wave solutions 
near the singularity is clarified in Section 3. The results of these two sections 
will accomplish the proof of Theorem \ref{theorem-main1}.

Section 4 describes the outcomes of the three numerical methods 
implemented for the approximation of solitary wave solutions 
of the differential equation (\ref{ode}) with $b = c = 1$. 
It is interesting that the regularization method 
approximates the bell-shaped soliton only, 
Petviashvili's method approximates the cusped soliton only, 
and Newton's method allows to approximate
both the bell-shaped and cusped solitons as well as 
other solutions in the continuous family of solitary waves.

Spectral stability of the solitary wave solutions 
is addressed in Section 5 
in the framework of the linearized NLS equation. We show how 
to characterize the kernel of the linearized operator and raise 
an open question on the mathematical analysis of the spectral stability problem. Numerical results suggest that the spectrum of the linearized operator is neutrally stable both for the cusped and bell-shaped solitons.

Finally, Section 6 summarizes our findings and presents
some directions of future study.

\section{Solitary wave solutions of the model}

We consider the differential equation (\ref{ode}) for $c > 0$.
The positive parameter $c$ can be set to unity without loss of generality
because if $u(x) = U(\sqrt{c}x)$ satisfies (\ref{ode}) for $c > 0$, 
then $U(x)$ satisfies the same equation with $c = 1$. 
Hence, we set $c = 1$ and rewrite the second-order equation (\ref{ode}) as the Newton equation:
\begin{align}
\label{Newton}
\frac{d^2 u}{dx^2} = \frac{u}{1-bu^2} = -V'(u),
\end{align}
where the potential $V$ is given by
\begin{align}
\label{potential}
V(u) = -\int \frac{u du}{1-bu^2} = \frac{1}{2b} \log|1-bu^2|.
\end{align}
The first invariant for the Newton equation (\ref{Newton}) is given by
\begin{align}
\label{ode-energy}
\frac{1}{2} \left(\frac{du}{dx} \right)^2 + V(u) = C,
\end{align}
where the value of $C$ is constant along every smooth solution of
the Newton equation (\ref{Newton}).

If $b \neq 0$, then it can be set to unity up to the choice of its sign
without loss of generality because if 
$u(x) = |b|^{-1/2} U(x)$ satisfies
(\ref{Newton}) for $b \neq 0$, then $U(x)$ 
satisfies the same equation with either $b = 1$
or $b = -1$. In what follows, we consider the two cases separately.

\subsection{Solitary wave solutions for $b = -1$}

We show that no solitary wave solutions exist in the Newton equation 
(\ref{Newton}) for $b = -1$ 
(or generally, for $b < 0$).

\begin{proposition}
	\label{prop-soliton-1}
	There exist no solutions with $u(x) \to 0$ as $|x| \to \infty$ 
in the Newton equation (\ref{Newton}) with $b = -1$.
\end{proposition}

\begin{proof} 
	If $b = -1$, the potential $V(u)$ can be written in the form:
\begin{align}
\label{potential-negative}
V(u) = -\frac{1}{2}\log(1+u^2).
\end{align}
All solutions are uniquely defined by the level $C$ in (\ref{ode-energy}) 
and remain smooth due to the smoothness of $V(u)$ in (\ref{potential-negative}). 
Solutions satisfying $u(x) \to 0$ as $|x| \to \infty$ 
correspond to the level $C = 0$ since 
$V(0) = 0$. They exist if and only if there exist nonzero 
turning points given by nonzero roots of $V(u)$. 
Since $V(u) < 0$ for every $u > 0$, 
no nonzero turning points exist at the level $C = 0$.
\end{proof}

\begin{figure}[hbt]
	\centering
	\includegraphics[width=12cm,height=8cm]{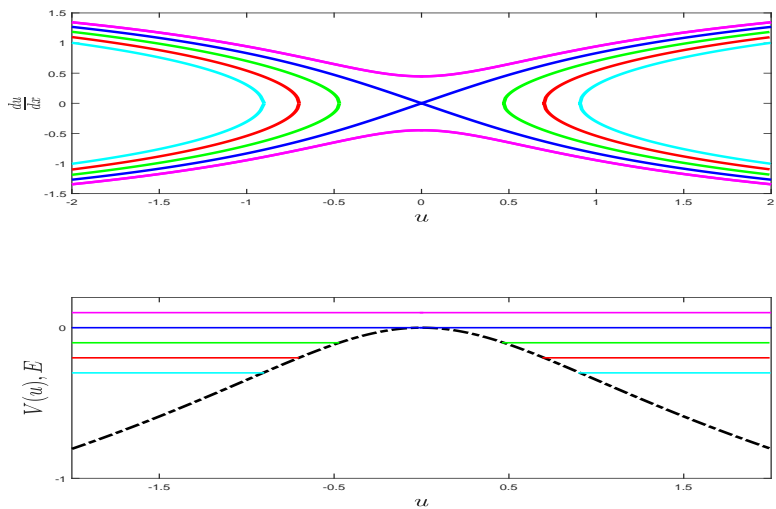}
	\caption{Top: Orbits on the phase plane for the potential (\ref{potential-negative}) corresponding to energy levels
		$C=0.2, 0, -0.1, -0.2, -0.3$. Bottom:
		levels of $C$ relative to the potential $V$ with admissible regions
		occurring for $V(u) \leq C$.}
	\label{ppandV_b=-1}
\end{figure}

Fig. \ref{ppandV_b=-1} (top panel) shows the level curves of the 
function in (\ref{ode-energy}) on the phase plane $(u,u')$.
The level curves with $C > 0$ ($C < 0$) lie outside (inside)
the stable and unstable curves corresponding to $C = 0$.
All curves are unbounded since no two turning points exist
for each orbit (see the bottom panel).

\begin{remark}
It was claimed in \cite{OL2020} that 
solitary wave solutions may exist for $b < 0$, in contradiction 
to Proposition \ref{prop-soliton-1}. The problem
with the approach of \cite{OL2020} stems from the Taylor expansion of the potential
$V(u)$ and truncation of this expansion. Indeed, the potential in  (\ref{potential-negative}) can be expanded as
\begin{align}
\label{potential-expansion}
V(u) = -\frac{1}{2} u^2 + \frac{1}{4} u^4 + \mathcal{O}(u^6) \quad \mbox{\rm as}
\quad u \to 0.
\end{align}
If the remainder term is truncated, the truncated Taylor expansion
(\ref{potential-expansion}) admits artificial turning points at $u = \pm \sqrt{2}$, which are not present in the original potential (\ref{potential-negative}). As a result, the truncated problem
has the artificial solution $u(x) = \sqrt{2} {\rm sech}(x)$ which does not
persist in the full system with the potential (\ref{potential-negative}). 
Note that further to the Taylor expansion (\ref{potential-expansion}), 
the approach of \cite{OL2020} used the expansion of the integrand near $u = 0$, after which the artificial solution was approximated with the Lambert-$W$ function. 
\end{remark}

Fig. \ref{lambertsoliton} shows evolution of the NLS equation (\ref{nls})
from the initial condition $\psi(x,0) = \sqrt{2} {\rm sech}(x)$.
This evolution leads to 
dispersion, corroborating the absence of a solitary wave.

\begin{figure}[hbt]
	\centering
  	\includegraphics[scale=0.45]{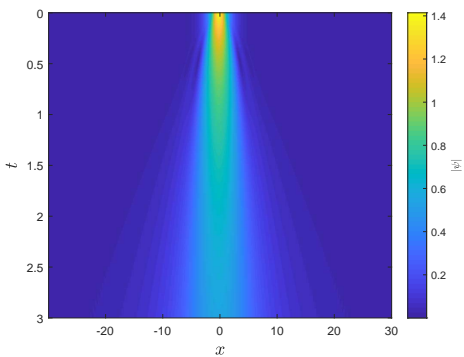}
	\caption{The space-time evolution of the NLS equation
          (\ref{nls}) with $\psi(x,0)=\sqrt{2} {\rm sech}(x)$. The evolution leads to spreading (dispersion) of the initially localized pulse as radiation is emitted.}
	\label{lambertsoliton}
\end{figure}

\subsection{Solitary wave solutions for $b = 1$}

We show that a continuous family of positive, single-humped, and continuous solitary wave solutions exists formally in the Newton equation (\ref{Newton}) for $b = 1$ 
(or generally, for $b > 0$).

\begin{proposition}
	\label{prop-soliton-2}
There exists a one-parameter family of positive, single-humped, and continuous 
solutions with $u(x) \to 0$ as $|x| \to \infty$ 
in the Newton equation (\ref{Newton}) with $b = 1$.
\end{proposition}

\begin{proof}
If $b = 1$, the potential $V(u)$ can be written in the form:
\begin{align}
\label{potential-positive}
V(u) = \frac{1}{2}\log|1-u^2|.
\end{align}
Two logarithmic singularities exist at $u=\pm 1$. Solutions of 
the Newton equation (\ref{Newton}) with $u(x) \to 0$ as $|x| \to \infty$ 
correspond to the level $C = 0$ since $V(0) = 0$. 
The turning points at the level $C = 0$ are $u = \pm \sqrt{2}$, 
hence the positive and negative solutions for $u(x)$ pass the singularities 
at $u = \pm 1$, at which the derivative $u'(x)$ becomes unbounded due 
to the first-order invariant (\ref{ode-energy}). 

A general way to continue the solution beyond the singularity with $u(x)$ being continuous through the breaking point is to concatenate 
the smooth solution for $u(x) < 1$ corresponding to the level $C = 0$ 
with another smooth solution for $u(x) \geq 1$ corresponding to 
an arbitrary level $C \in \mathbb{R}$. This gives the one-parameter family 
of solutions parametrized by $C \in \mathbb{R}$ for the part of the 
solution with $u(x) \geq 1$.
\end{proof}

Within the one-parameter family of solitary wave solutions of Proposition \ref{prop-soliton-2}, we define two particular solutions:
\begin{itemize}
	\item The cusped soliton (left panel of Fig. \ref{bluesolns}), 
	which has the infinite jump singularity for $u'(x)$. It formally corresponds 
	to $C = -\infty$ for the part of the solution with $u(x) \geq 1$.\\
	
	\item The bell-shaped soliton  (right panel of Fig. \ref{bluesolns}), 
	which has the same infinite value of the first derivative at the two singularities. 
	It formally corresponds to $C = 0$ for the part of the solution with $u(x) \geq 1$.
\end{itemize}

\begin{remark}
	It was claimed in \cite{OL2020} that one positive, single-humped solitary wave solution may exist for $b > 0$ as the bell-shaped soliton. 
	Proposition \ref{prop-soliton-2} alludes to a continuous family of positive, single-humped solitary waves.
\end{remark}

The level curves and energy levels $C$ for the potential $V(u)$ in (\ref{potential-positive}) are shown on Fig. \ref{ppandV_b=1}.
Other solutions constructed from the first-order invariant
(\ref{ode-energy}) beyond the singularity at $u = \pm 1$ are very
similar, i.e., they feature similar ways of continuing past the singularity.

For $C < 0$, the solutions are periodic and (can be thought of as
being) positive definite as $u(x)$ is squeezed between the turning points.
The two periodic solutions (cusped and bell-shaped) are shown on Fig. \ref{magentasoln} with the same values of $C$ below and above the singularity at $u = 1$. As $C \to 0$, these two periodic solutions 
become the cusped and bell-shaped solitons since their periods diverge to infinity. 

For $C > 0$, the periodic solutions become double-humped with the alternating polarities. At each period, the solution reaches both singularity points
$u = \pm 1$. Therefore, there exist four ways to
define the double-humped periodic solutions with the same value of $C$ along each smooth piece of the solution. Three of the solutions are shown in Figure \ref{doublehump}. One more solution is identical to the solution on the left panel due to the transformation $u \mapsto -u$ for solutions of the Newton equation (\ref{Newton}). 

\begin{figure}[hbt]
	\centering
	\includegraphics[width=12cm,height=8cm]{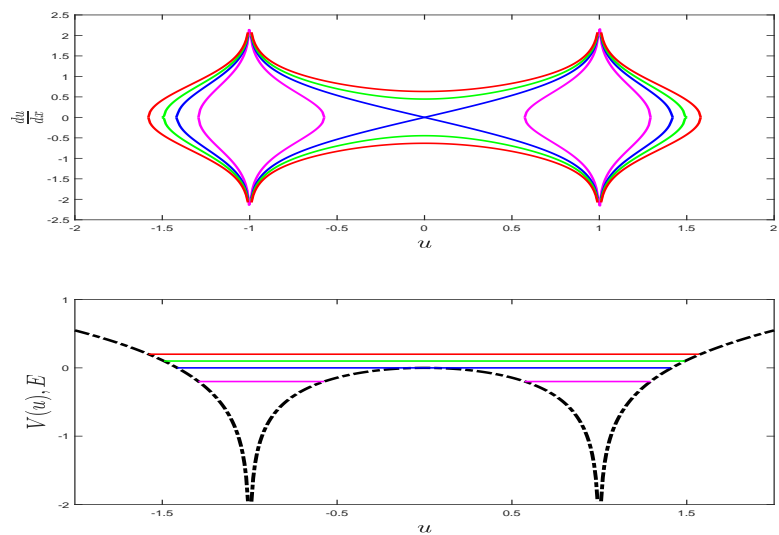}
	\caption{The same as Figure \ref{ppandV_b=-1} but for the potential
		(\ref{potential-positive}) and the energy levels $C = -0.2, 0, 0.1, 0.2, 0.3$.}
	\label{ppandV_b=1}
\end{figure}

\begin{figure}[hbt]
	\centering
	\includegraphics[scale=0.3]{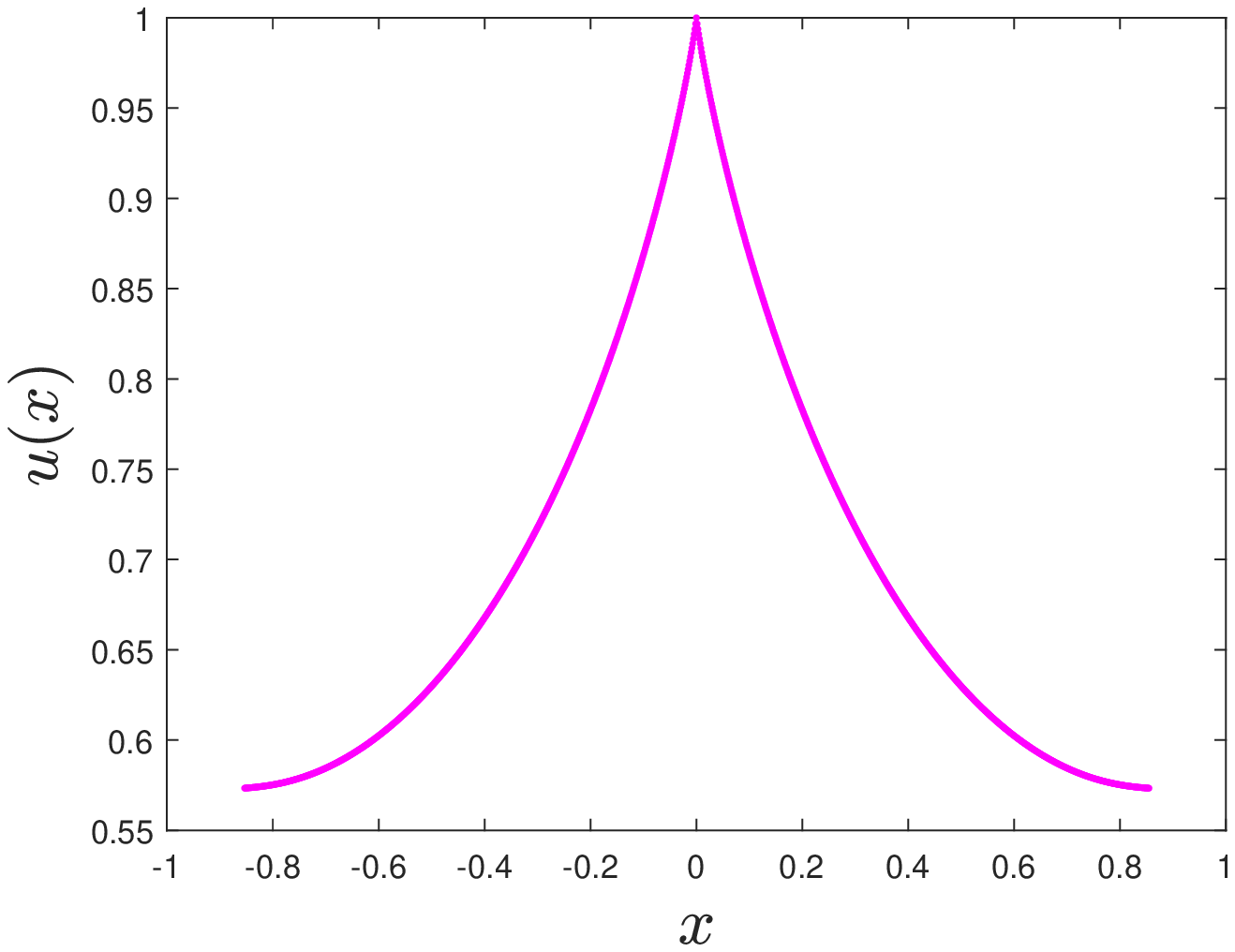}
	\includegraphics[scale=0.3]{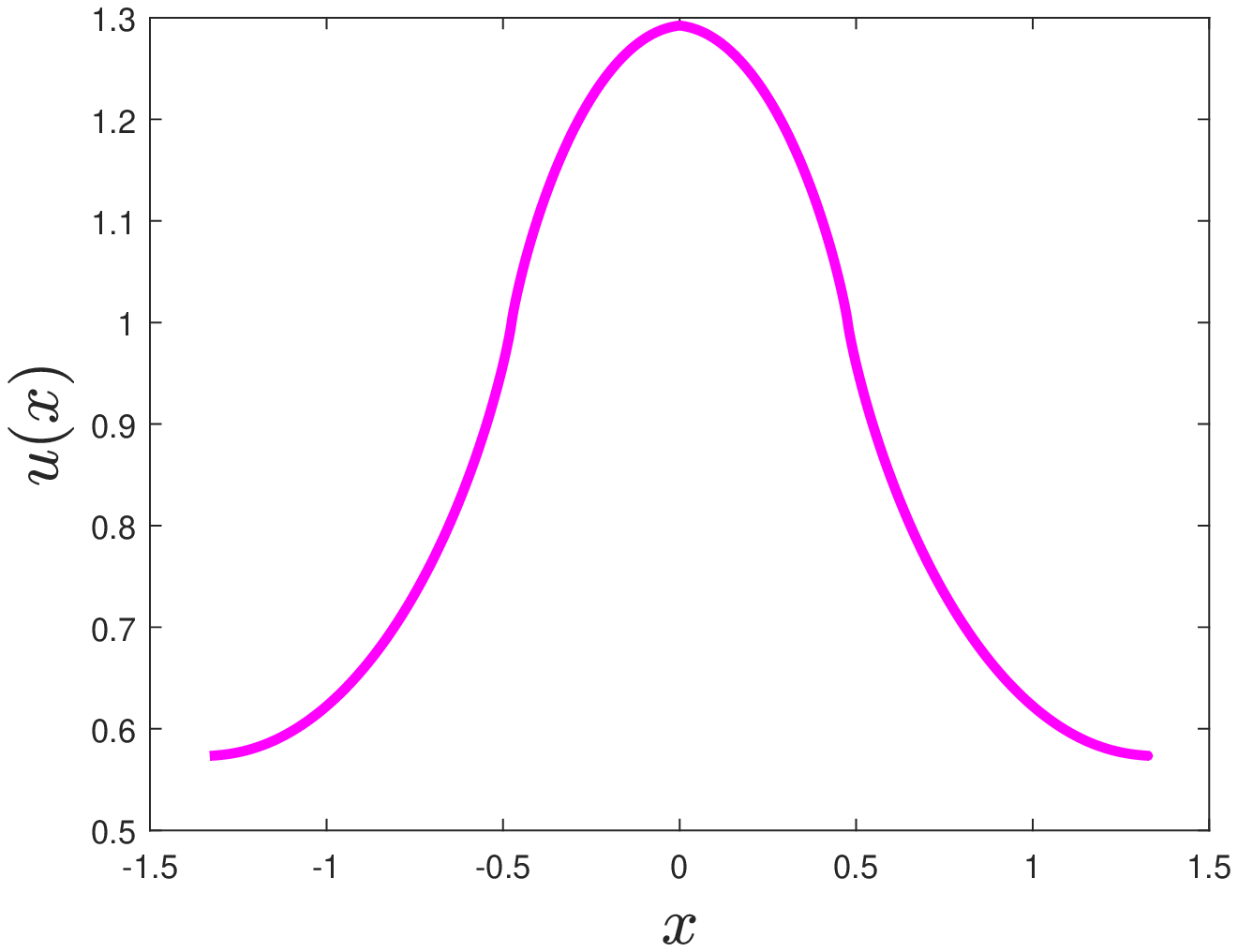}
	\caption{
		The spatial profiles of positive periodic solutions
	 for $C = -0.2$.
	}
	\label{magentasoln}
\end{figure}

\begin{figure}[hbt]
	\centering
	\includegraphics[scale=0.3]{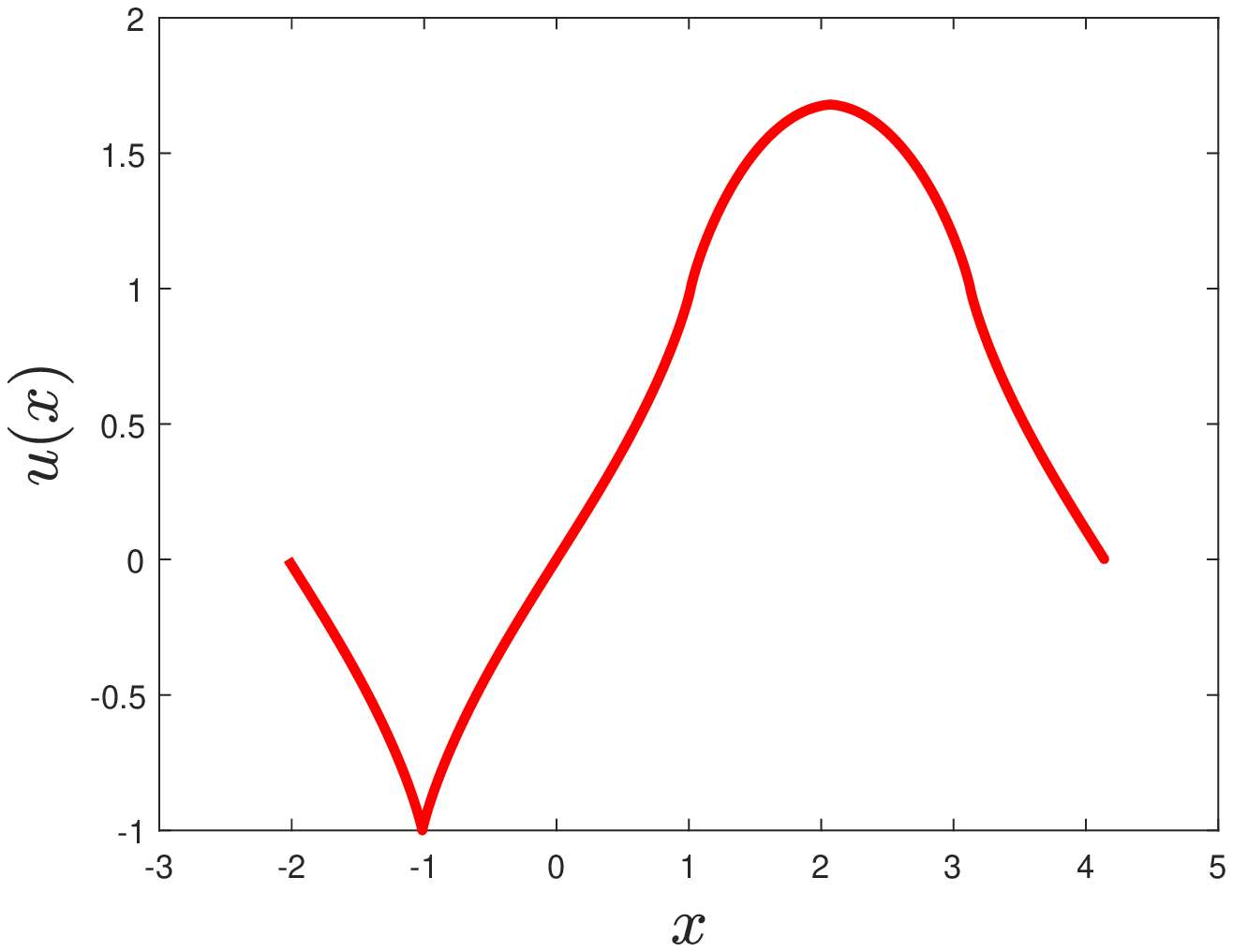}
	\includegraphics[scale=0.3]{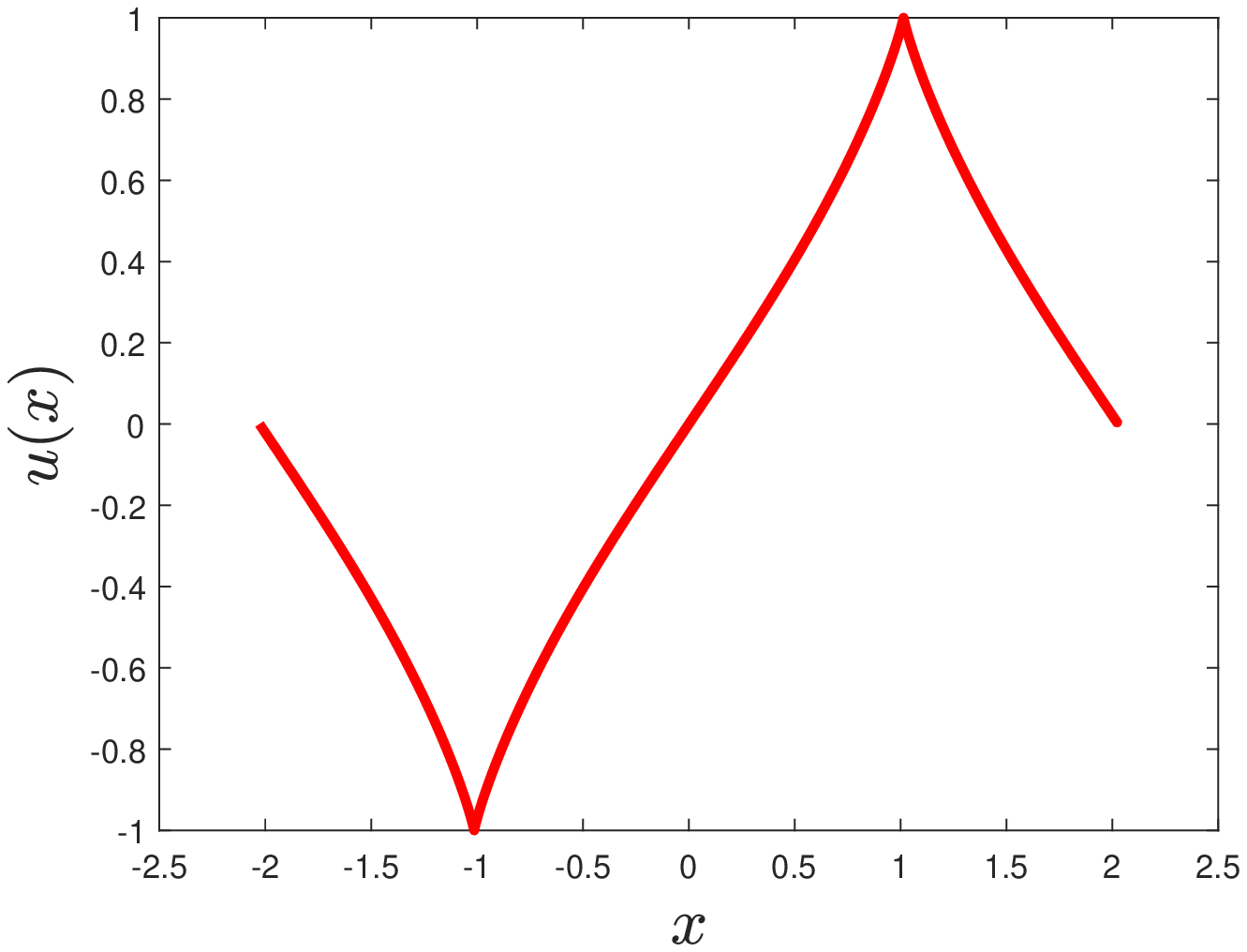}
	\includegraphics[scale=0.3]{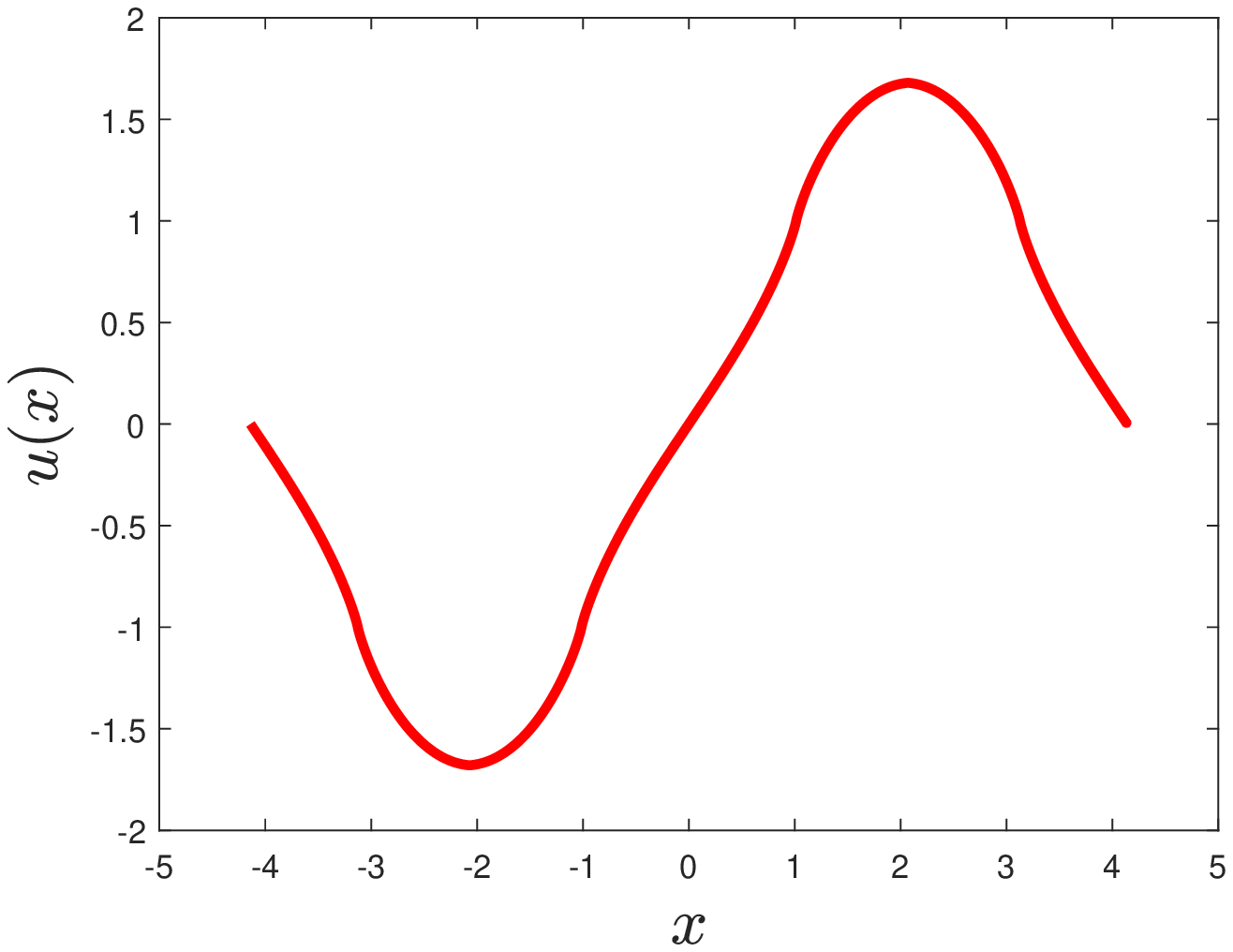}
	\caption{The spatial profiles of sign-definite periodic solutions for $C=0.3$.}
	\label{doublehump}
\end{figure}

\section{Singular behavior near the logarithmic singularity}
\label{sec-singularity}

Although we have formally obtained a one-parameter family of positive, single-humped solitary wave solutions in Proposition \ref{prop-soliton-2}, it remains to justify the existence of such solutions in the weak formulation (\ref{ode-weak}) 
with $c = 1$ and $b = 1$. We do so by clarifying the singular behavior of positive solutions near the logarithmic singularity at $u = 1$ 
and by verifying that the solitary wave solutions belong to $H^1(\mathbb{R})$.

The Newton equation (\ref{Newton}) with $b = 1$ can be rewritten in the form:
\begin{equation}
u''(x) = \frac{u(x)}{1 - u(x)^2}.
\label{2nd}
\end{equation}
Let $u_{\rm cusp}$ denote the cusped soliton. The cusped soliton $u_{\rm cusp}$ is defined by the implicit equation that follows from integration of the first-order invariant (\ref{ode-energy}) with $C = 0$:
\begin{align}
|x-x_0| = \int_{u}^1 \frac{d\xi}{\sqrt{-\log(1-\xi^2)}}, \quad u \in (0,1),
\label{0th}
\end{align}
where $x_0 \in \mathbb{R}$ is arbitrary due to the translational symmetry.
Without loss of generality, we place the cusped soliton $u_{\rm cusp}$ at the origin by selecting $x_0 = 0$ in (\ref{0th}).

Let $u_{\rm bell}$ denote the bell-shaped soliton defined piecewise as follows:
\begin{equation}
\label{bell}
u_{\rm bell}(x) = \left\{ \begin{array}{ll} u_{\rm head}(x), \quad & x \in [-\ell,\ell], \\
u_{\rm cusped}(|x| - \ell), \quad & |x| > \ell, \end{array} \right.
\end{equation}
where $\ell$ is uniquely defined by
\begin{equation}
\label{def-L}
\ell := \int_1^{\sqrt{2}} \frac{du}{\sqrt{|\log(u^2-1)|}}
\end{equation}
and $u_{\rm head}(x) \in [1,\sqrt{2}]$ for $x \in [-\ell,\ell]$ is defined
implicitly by
\begin{equation}
\ell-|x| = \int_1^u \frac{d\xi}{\sqrt{-\log(\xi^2-1)}}, \quad u \in (1,\sqrt{2}].
\label{bell-head}
\end{equation}

Finally, the one-parameter family of solitary wave solutions 
in Proposition \ref{prop-soliton-2} is defined piecewise as follows:
\begin{equation}
\label{cont-family}
u_C(x) = \left\{ \begin{array}{ll} u_{{\rm head},C}(x), \quad & x \in [-\ell_C,\ell_C], \\
u_{\rm cusped}(|x| - \ell_C), \quad & |x| > \ell_C, \end{array} \right.
\end{equation}
where $\ell_C$ is uniquely defined by
\begin{equation}
\label{def-L-C}
\ell_C := \int_1^{\sqrt{1+e^{2C}}} \frac{du}{\sqrt{2C - \log(u^2-1)}}
\end{equation}
and $u_{{\rm head},C}(x)$ for $x \in [-\ell_C,\ell_C]$ is defined
implicitly by
\begin{equation}
\ell_C-|x| = \int_1^u \frac{d\xi}{\sqrt{2C-\log(\xi^2-1)}}, \quad u \in (1,\sqrt{1+e^{2C}}].
\label{bell-head-continuous}
\end{equation}
If $C = 0$, then $u_{C=0} \equiv u_{\rm bell}$ with $\ell_{C=0} \equiv \ell$.
If $C = -\infty$, then $u_{C = -\infty} \equiv u_{\rm cusp}$ with $\ell_{C=-\infty} \equiv 0$.

The following proposition gives the asymptotic behavior of 
$u_{\rm cusp}$ near the logarithmic singularity 
at $u = 1$. The proof follows closely the proof of Lemma 2.4 in \cite{Alfimov}.

\begin{proposition}
	\label{prop-cusped}
	Let $u_{\rm cusp}$ be the cusped soliton given by the implicit
	equation (\ref{0th}) with $x_0 = 0$. Then,
	\begin{align}
	u_{\rm cusp}(x) = 1 - |x| \sqrt{\log(1/|x|)} \left[ 1 + \mathcal{O}\left(\frac{\log\log(1/|x|)}{\log(1/|x|)}\right) \right], \quad \mbox{\rm as} \quad |x| \to 0,
	\label{asympt-sol}
	\end{align}
	where $\mathcal{O}(v)$ denotes a $C^1$ function of $v$ near $v = 0^+$.
\end{proposition}

\begin{proof}
We make the substitution $u = 1 - v$ and expand the integral in (\ref{0th}) with $x_0 = 0$ as follows:
	\begin{align}
	\nonumber
	|x| & = \int_{0}^v \frac{d\eta}{\sqrt{|\log(\eta)| \left(1 + \frac{\log(2-\eta)}{\log(\eta)}\right)}} \\
	& =
	\int_{0}^v \frac{d\eta}{\sqrt{|\log(\eta)|}} \left[ 1 + \mathcal{O}\left(\frac{1}{|\log(\eta)|}\right) \right] \quad \mbox{\rm as} \quad v \to 0^+.
	\label{0th-expand}
	\end{align}
	Since 
	$$
	\frac{d}{dv} \left[ \frac{v}{\sqrt{|\log(v)|}} \right] = \frac{1}{\sqrt{|\log(v)|}}
	+ \frac{1}{2 \sqrt{|\log(v)|^3}}, \quad v \in (0,1),
	$$
	we obtain from (\ref{0th-expand}) by integration by parts:
	\begin{align}
	|x| = \frac{v}{\sqrt{|\log(v)|}} \left[ 1 + \mathcal{O}\left(\frac{1}{|\log(v)|}\right)
	\right] \quad \mbox{\rm as} \quad v \to 0^+.
	\label{0th-integral}
	\end{align}
	Setting $v(x) = |x| \sqrt{|\log|x||} w(x)$ into (\ref{0th-integral}) yields
	the nonlinear equation
	\begin{align}
	w(x) = \sqrt{1 + \frac{\log|\log|x|| + 2 \log(w)}{2 \log|x|}} \left[ 1 + \mathcal{O}\left(\frac{1}{|\log(x)|}\right)
	\right] \quad \mbox{\rm as} \quad x \to 0,
	\label{0th-integral-w}
	\end{align}
	from which the existence and uniqueness of the root
	$w(x) = 1 + \mathcal{O}(\frac{\log|\log|x||}{|\log|x||})$ as $x \to 0$ is proved with
	the implicit function theorem since all correction terms are $C^1$ functions of $x$ and $w$. Substituting all transformations back gives the asymptotic expansion (\ref{asympt-sol}).
\end{proof}

\begin{remark}
With a similar transformation for the integral in (\ref{bell-head}), 
one can show that the bell-shaped soliton $u_{\rm bell}$  given by (\ref{bell}), (\ref{def-L}), and (\ref{bell-head}) admits the behavior
\begin{align}
u_{\rm bell}(x) = 1 + (\ell-|x|) \sqrt{|\log|\ell-|x|||} \left[ 1 + \mathcal{O}\left(\frac{\log|\log|\ell-|x|||}{|\log|\ell-|x|||}\right) \right], \quad \mbox{\rm as} \quad |x| \to \ell.
\label{asympt-sol-bell}
\end{align}
Similarly, the one-parameter family $u_C$ of solitary wave solutions given by (\ref{cont-family}), (\ref{def-L-C}), and (\ref{bell-head-continuous}) 
admits the behavior
\begin{align}
\nonumber
u_C(x) &= 1 + (\ell_C-|x|) \sqrt{|\log|\ell_C-|x|||} \left[ 1 + \mathcal{O}\left(\frac{\log|\log|\ell_C-|x|||}{|\log|\ell_C-|x|||}\right) \right. \\
 & \qquad \left.
+ \mathcal{O}_C\left(\frac{1}{|\log|\ell_C-|x|||}\right) \right], \quad \mbox{\rm as} \quad |x| \to \ell_C,
\label{asympt-sol-bell-cont}
\end{align} 
where $\mathcal{O}_C$ denotes remainder terms that depend on parameter $C \in \mathbb{R}$ for $x \in [-\ell_C,\ell_C]$.
\end{remark}

Figure \ref{fig-sol-asympts} shows a very good agreement of the two solutions  $u_{\rm cusp}$ and $u_{\rm bell}$ with their leading order approximations given by \eqref{asympt-sol} and \eqref{asympt-sol-bell}. 

\begin{figure}[ht]
	\includegraphics[width=8cm,height = 6cm]{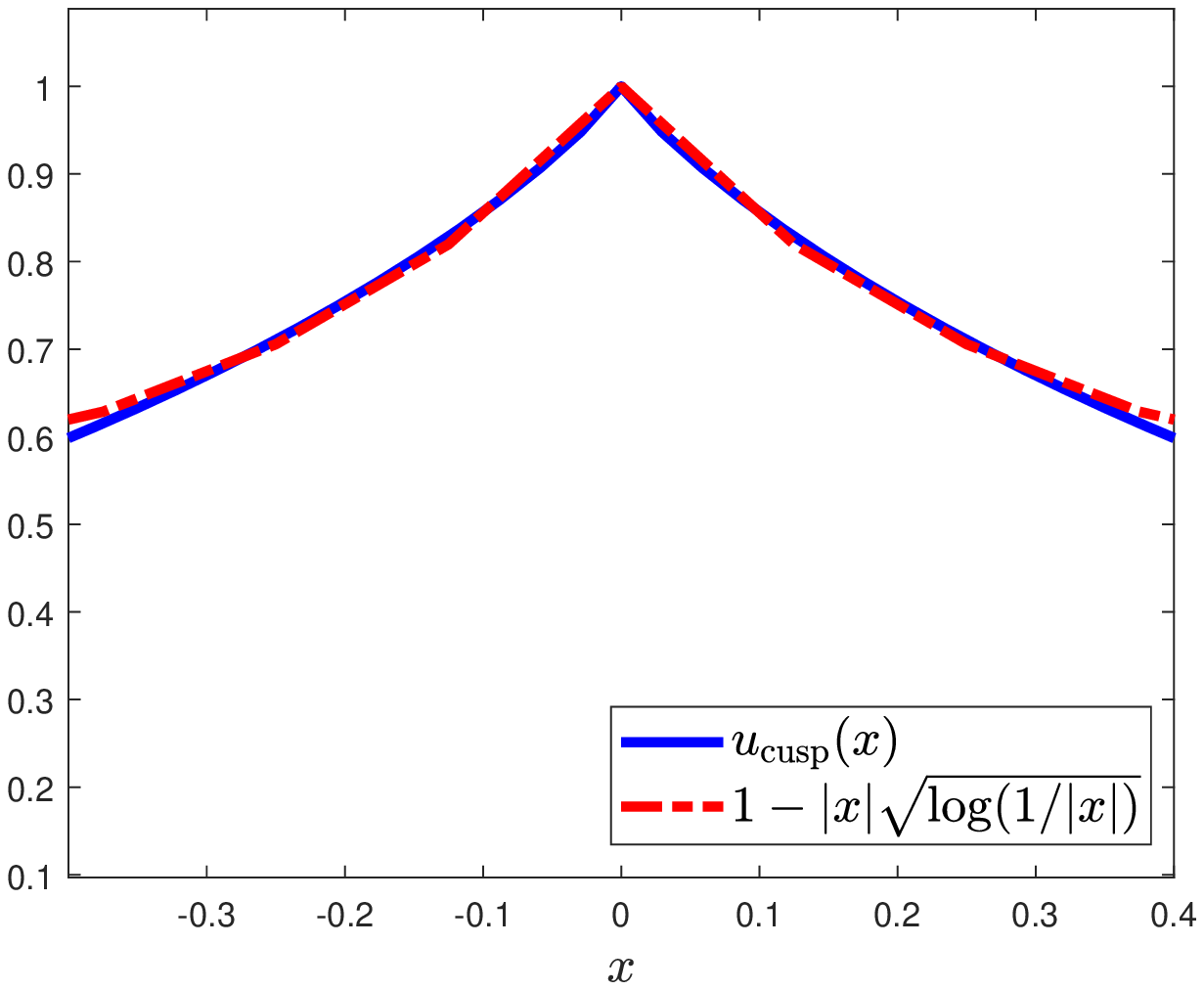}
	\includegraphics[width=8cm,height = 6cm]{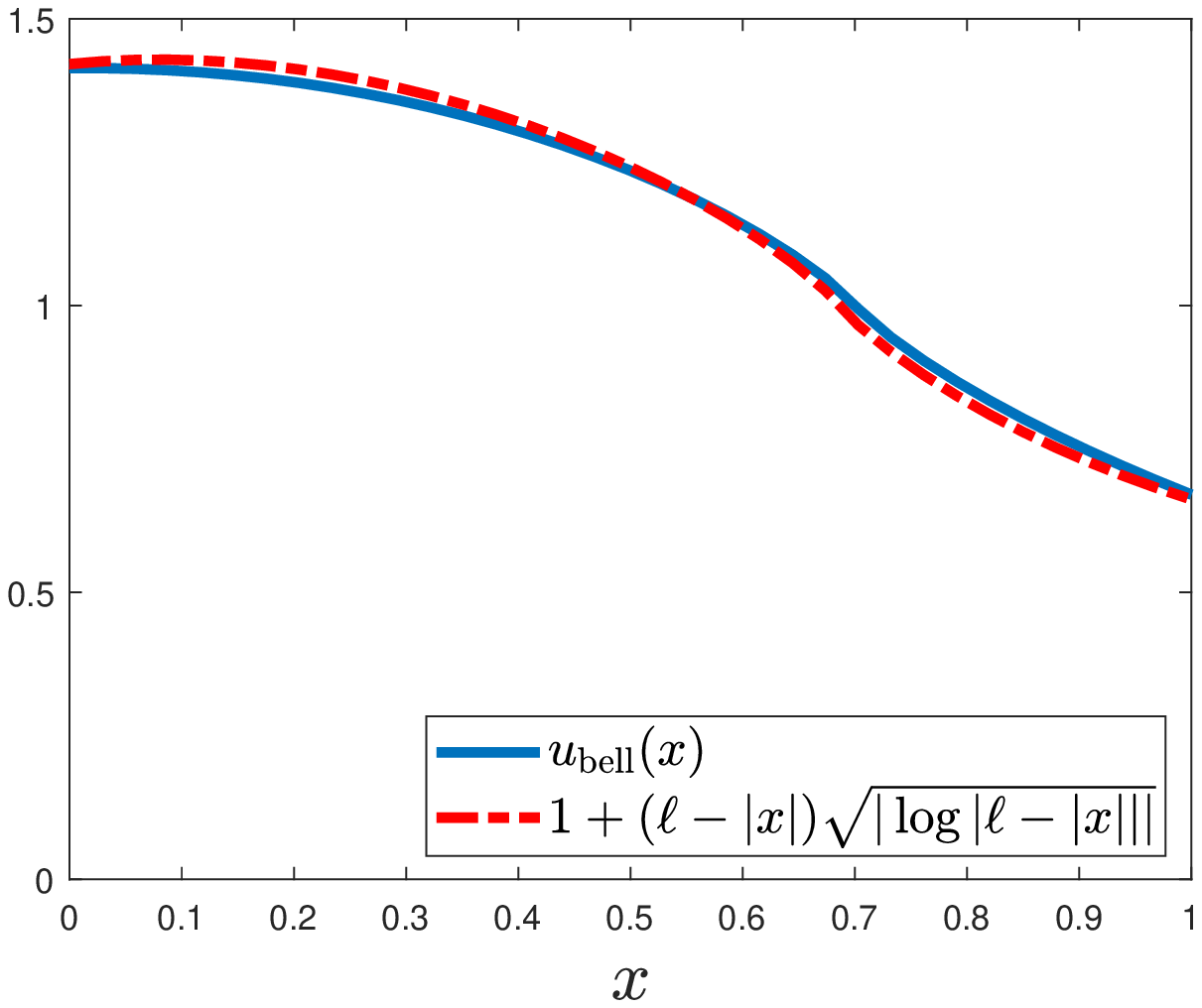}
	\caption{The leading-order approximations given by \eqref{asympt-sol} and \eqref{asympt-sol-bell} superposed with the numerically detected
		cusped (left) and bell-shaped (right) solitary waves.
	}
	\label{fig-sol-asympts}
\end{figure}

The solutions $u_{\rm cusp}$ and $u_{\rm bell}$ in Figures \ref{bluesolns} and \ref{fig-sol-asympts} are obtained numerically as follows. Since the cusped and bell-shaped solitons are even, we solve the implicit equations (\ref{0th}) and (\ref{bell-head}) for $x > 0$ and obtain the other half by the symmetry. To solve the integral equations, we discretize the computational domain $[0,L]$, and approximate the relevant integrals by the midpoint rule on the grid. This yields a nonlinear system of equations for the values of the solution $u$ at grid points, which is then solved using Newton's method.
 
For the cusped soliton, we solve \eqref{0th} with the method described above. For the bell-shaped soliton, we first obtain the solution on $[0,\ell]$, by solving \eqref{bell-head} for $u_{\rm head}(x)$ in the same way. 
The constant $\ell$ is computed from the integral \eqref{def-L} as $\ell \approx 0.6862$. Then, we construct the entire bell-shaped soliton according to \eqref{bell}. For $|x|>\ell$, the solution is defined using the shifted cusped soliton, so we use the cusped soliton already obtained from solving \eqref{0th}. \\

\vspace{0.2cm}

We are now ready to prove Theorem \ref{theorem-main1}. 
By Proposition \ref{prop-soliton-2}, a one-parameter family of positive and single-humped solitary wave solutions of the second-order equation (\ref{2nd}) exists. The solutions are continuous and decay to zero as $|x| \to \infty$ exponentially fast. By Proposition \ref{prop-cusped}, 
$u'(x)$ has infinite jump singularities but the singularities are weak 
so that $u_{\rm cusp}, u_{\rm bell}, u_C \in H^1(\mathbb{R})$. Moreover, 
$u_{\rm cusp}, u_{\rm bell}, u_C \in X \subset H^1(\mathbb{R})$. Each smooth part of the solution 
in $u_{\rm cusp}$,  $u_{\rm bell}$, and $u_C$ 
satisfies the weak formulation in (\ref{ode-weak}) 
for $c = 1$ and $b = 1$ with compactly supported test functions $\varphi$ in appropriate regions of $\mathbb{R}$. The weak formulation in Definition \ref{def-weak-solution} does not 
impose any jump conditions on derivatives of $u$ at the breaking points where $u = 1$.
The proof of Theorem \ref{theorem-main1} is complete.

\section{Numerical methods for solitary wave solutions}

Here we study convergence of the three numerical methods 
used to obtain solitary wave solutions in the differential 
equation (\ref{ode}) with $b = c = 1$, which is also written 
as  (\ref{2nd}).

\subsection{Bell-shaped soliton via regularization}

A natural regularization of the singular second-order equation (\ref{2nd})
is given by 
\begin{align}
u_\eps'' = \frac{u_\eps(1-u_\eps^2)}{(1-u_\eps^2)^2+\eps^2},
\label{reg2}
\end{align}
where $\eps > 0$ is a small parameter. The formal limit $\eps \to 0$
recovers (\ref{2nd}). The first-order invariant for the regularized
equation (\ref{reg2}) is given by 
\begin{align}
\label{reg1}
\frac{1}{2} \left( \frac{d u_\eps}{dx} \right)^2 + V_\eps(u_\eps)=C
\end{align}
with the potential $V_{\eps}(u)$ given by
\begin{align}
V_\eps(u) = \frac{1}{2}\log\bigg[\frac{\sqrt{(1-u^2)^2+ \eps^2}}{\sqrt{1+\eps^2}}\bigg],
\end{align}
where the denominator ensures that the critical point $(0,0)$ still corresponds
to the level $C = 0$. Figure \ref{fig-reg} shows the level curves of the regularized first-order invariant (\ref{reg1}). Figure \ref{H1_conv} shows the profiles of the bell-shaped soliton for different values of $\eps > 0$ (left) 
and illustrates the convergence $u_\eps \to u_{\rm bell}$ in the
$H^1(\R)$ norm as $\eps\to 0$ (right). The following proposition justifies these numerical results analytically.

\begin{proposition}
	\label{prop-regularization}
	For every $\eps > 0$, there exists only one smooth positive solitary wave solution $u_{\eps}$ of the second-order equation (\ref{reg2}) such that $0 < u_{\eps}(x) \leq \sqrt{2}$. Moreover, 
	$$
	\| u_{\eps} - u_{\rm bell}\|_{H^1} \to 0 \quad \mbox{\rm as} \quad \eps \to 0.
	$$
\end{proposition}

\begin{proof}
	The second-order equation (\ref{reg2}) and its first-order invariant (\ref{reg1}) are smooth for every $u \in \mathbb{R}$ if $\eps > 0$. 
	The positive solitary wave solution corresponds to the level $C = 0$, for which the turning point is located at $u = \sqrt{2}$ for every $\eps > 0$. The positive solitary wave solution is defined up to the translation in $x$ by the implicit equation:
\begin{align}
\label{reg-int}
|x| = \int_u^{\sqrt 2} \frac{d\xi}{\sqrt{-2V_\eps(\xi)}}, \qquad 
u \in (0,\sqrt 2),
\end{align}
where the integrand has a weak singularity at $\xi = \sqrt{2}$ and is smooth for any $\xi \in (u,\sqrt{2})$. This gives $u_{\eps} \in C^{\infty}(\mathbb{R})$ 
satisfying $0 < u_{\eps}(x) \leq \sqrt{2}$. Since 
$V_{\eps}(x) \to V(x)$ as $\eps \to 0$ for every $x \in \mathbb{R}$ 
and $|V|^{-1/2}, |V_{\eps}|^{-1/2} \in L^1(u,\sqrt{2})$ for every $u \in (0,\sqrt{2}]$, Lebesgue's dominated convergence theorem implies 
that $u_{\eps}(x) \to u_{\rm bell}(x)$ as $\eps \to 0$  for every $x \in \mathbb{R}$. Because $u_{\eps}(x), u_{\rm bell}(x) \to 0$ as $|x| \to \infty$ 
exponentially fast  with the same rate, the pointwise convergence implies that $\| u_{\eps} - u_{\rm bell} \|_{L^2} \to 0$ as $\eps \to 0$. Since $u_{\eps}', u_{\rm bell}' \in L^2(\mathbb{R})$, the first-order invariant (\ref{reg1}) with $C = 0$ implies $\| u_{\eps}' \|_{L^2} \to \| u_{\rm bell}' \|_{L^2}$ as $\eps \to 0$, 
which yields $\| u_{\eps}' - u_{\rm bell}' \|_{L^2} \to 0$ as $\eps \to 0$. Hence $\| u_{\eps} - u_{\rm bell} \|_{H^1} \to 0$ as $\eps \to 0$.
\end{proof}

\begin{remark}
The implicit equation \eqref{reg-int} was solved numerically with Newton's method for some $x \in \mathbb{R}$ in order to obtain the bell-shaped soliton $u_{\eps}$ shown on Fig. \ref{H1_conv}.
\end{remark}

\begin{figure}[htpb]
	\centering
	\includegraphics[scale=0.7]{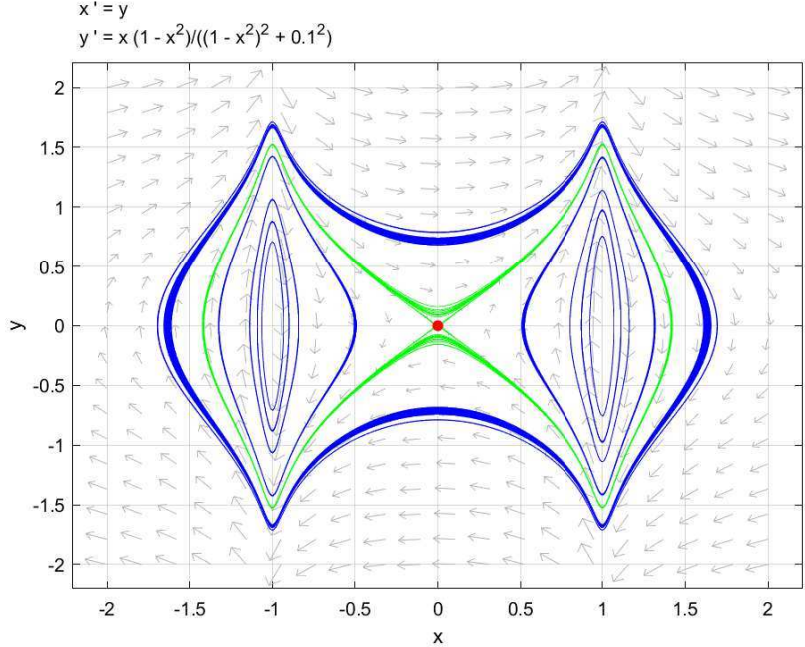}
	\caption{Phase portrait for the regularized equation \eqref{reg2} with $\eps=0.1$. The bell-shaped soliton is shown by green line.}
	\label{fig-reg}
\end{figure}

\begin{figure}[htpb]
	\includegraphics[width=8cm,height = 6cm]{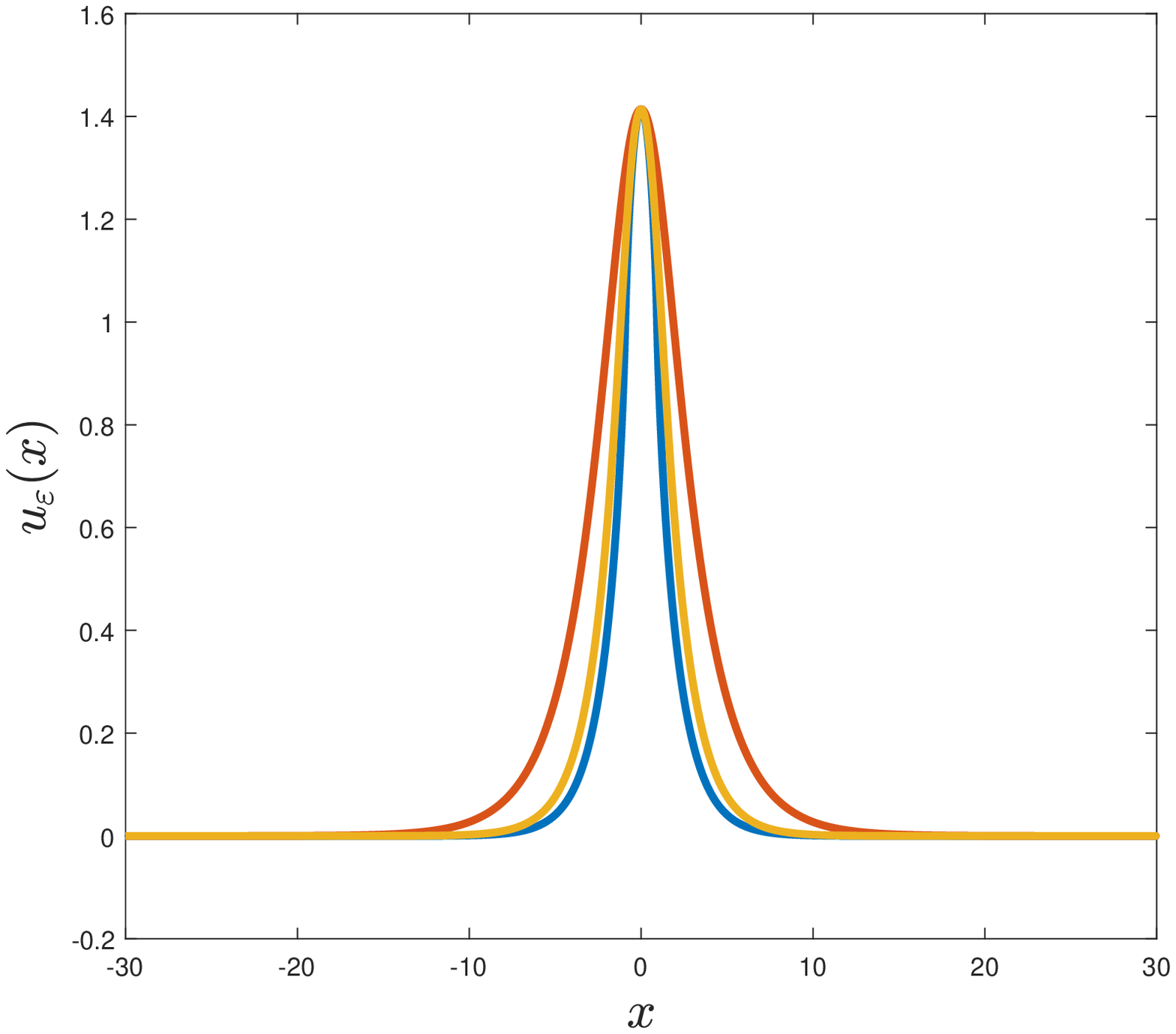}
	\includegraphics[width=8cm,height = 6cm]{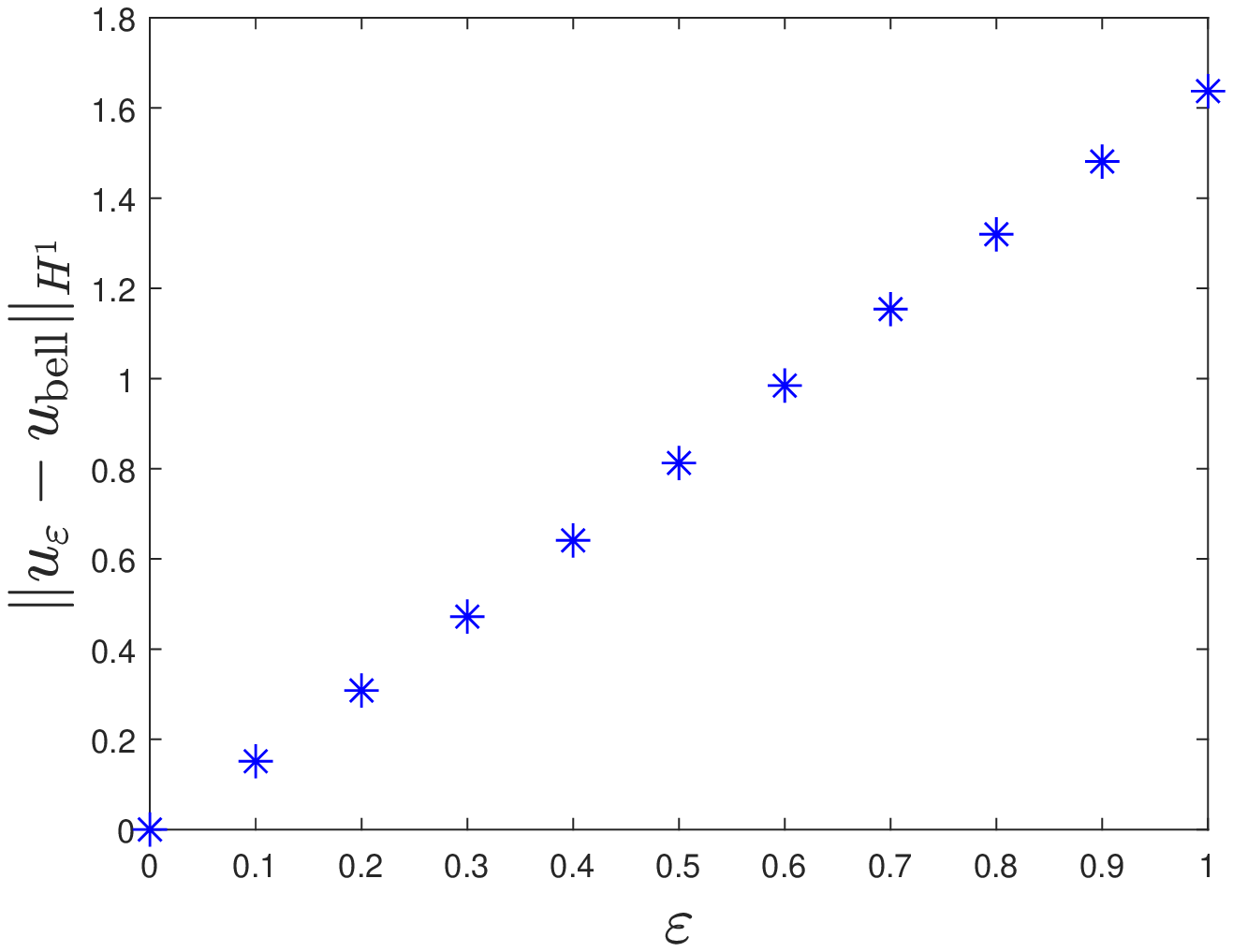}
	\caption{The spatial profile of the  bell-shaped soliton
		$u_\eps(x)$
		of the regularized equation \eqref{reg2} for $\eps = 2,1,0$ (left). Convergence of $u_{\eps}$ to $u_{\rm bell}$
		in the $H^1(\R)$ norm as $\eps\to 0$ (right).}
	\label{H1_conv}
\end{figure}

\subsection{Cusped soliton via Petviashvili's method}

We rewrite the stationary equation (\ref{2nd}) into the following equivalent form:
\begin{equation}
\label{seq-1}
u = (1-u^2) u'' \quad \Rightarrow u - u'' = -u^2 u''.
\end{equation}
A solution $u$ to the stationary equation (\ref{seq-1}) in $H^1(\mathbb{R})$ 
is a fixed point $u = T(u)$ of the nonlinear operator
\begin{equation}
\label{seq-2}
T(u) := -(1 - \partial_x^2)^{-1} u^2 \partial_x^2 u.
\end{equation}
Furthermore, a solution $u \in H^1(\mathbb{R})$ to the stationary equation (\ref{seq-1}) satisfies the equality
\begin{equation}
\label{seq-3}
\int_{\mathbb{R}} (u^2 + (u')^2) dx = 3 \int_{\mathbb{R}} u^2 (u')^2 dx,
\end{equation}
which follows from the weak formulation (\ref{ode-weak}) with $\varphi = u$ 
and $b = c = 1$.

Let us define the iterative method for $\{ w_n \}_{n \in \mathbb{N}} \in H^1(\mathbb{R})$ by
\begin{equation}
\label{seq-4}
w_{n+1} = -\lm_n^{3/2} (1 - \partial_x^2)^{-1} w_n^2 \partial_x^2 w_n,
\end{equation}
starting from an initial guess $w_0$, where $\lm_n := \lm(w_n)$ is the normalization constant defined by
\begin{equation}
\label{seq-5}
\lm(w) = \frac{\int_{\mathbb{R}} (w^2 + w_x^2) dx}{3 \int_{\mathbb{R}} w^2 w_x^2 dx}.
\end{equation}
The special power of $\lm_n$ is introduced in such a way that if $w_n = a_n u$, where $u \in H^1(\mathbb{R})$ is the true solution of the stationary equation
(\ref{seq-1}) satisfying the equality (\ref{seq-3}), then the
iterative method (\ref{seq-4})--(\ref{seq-5}) yields $\lambda_n =
a_n^{-2}$ and $a_{n+1} = 1$, so iterations converge after the first
step independently of $a_0 \neq 0$.

As the method proceeds, $\lm_n$ is supposed to converge to 1 and the sequence $\{w_n\}$ should converge to an approximate solution of the stationary equation \eqref{seq-1}. Hence, we measure convergence of the iterations by $|1-\lm_n|$ and $\|e_n\|_{L^\infty}$, where $e_n(x) = (1- w_n(x)^2)w_n''(x)-w_n(x)$. We stop the iterations at step $N$ when the convergence criterion $\|e_N\|_{L^\infty}<10^{-10}$ is reached. 

To compute all spatial derivatives, we use Fourier spectral differentiation matrices as follows, see, e.g., ~\cite{Tr-book}. For the normalized interval $[0,2\pi]$, we work on the grid
\begin{align}
\label{grid}
x_j=jh, \quad j \in \{1,\dots, N\}
\end{align}
where $N$ is a pre-chosen (large) even integer and $h = \frac{2\pi}{N}$ is the grid spacing. The left endpoint $0$ is removed so that the grid has exactly $N$ points. The choice of which endpoint to remove can be made arbitrarily, as the differentiation matrices are the same regardless of which endpoint is removed.

For the truncated interval $[-L,L]$, we need to translate and rescale the starting interval $[0,2\pi]$ by using the transformation 
\begin{equation}
x\mapsto y = \frac{L}{\pi}(x-\pi)
\end{equation} 
First and second order differentiation for functions on $[-L,L]$ 
on the grid points with grid spacing $L h/\pi$ is performed using the circulant matrices from \cite{Tr-book}.

Fig. \ref{fig-cusped} (left) shows how $|1-\lm_n|$ and $\| e_n \|_{L^{\infty}}$ converge in $n$ for the iterations
of the method (\ref{seq-4})--(\ref{seq-5}) with the initial guess
$w_0(x) = \sech(x)$ and $c = 1$. The algorithm was terminated after $N=287$ iterations when the aforementioned tolerance was reached.
Fig. \ref{fig-cusped} (right) shows that the iterations converged to the cusped soliton. The graph of the error shows that the error is maximal at the point of singularity at $x = 0$ with $\max_{x\in \R} |e_N(x)| < 10^{-10}$.

Trying the initial guess $w_0(x) = \sqrt{2}\,  {\rm\sech}(x)$ in an attempt to converge to the bell-shaped soliton, we can see clearly in Fig. \ref{fig-sqrt2-in-p} that the method still converges to the cusped soliton. This computation suggests that the bell-shaped soliton is unstable in the iterations
of the Petviashvili method (\ref{seq-4})--(\ref{seq-5}).

\begin{figure}[htb]
	\includegraphics[scale=0.45]{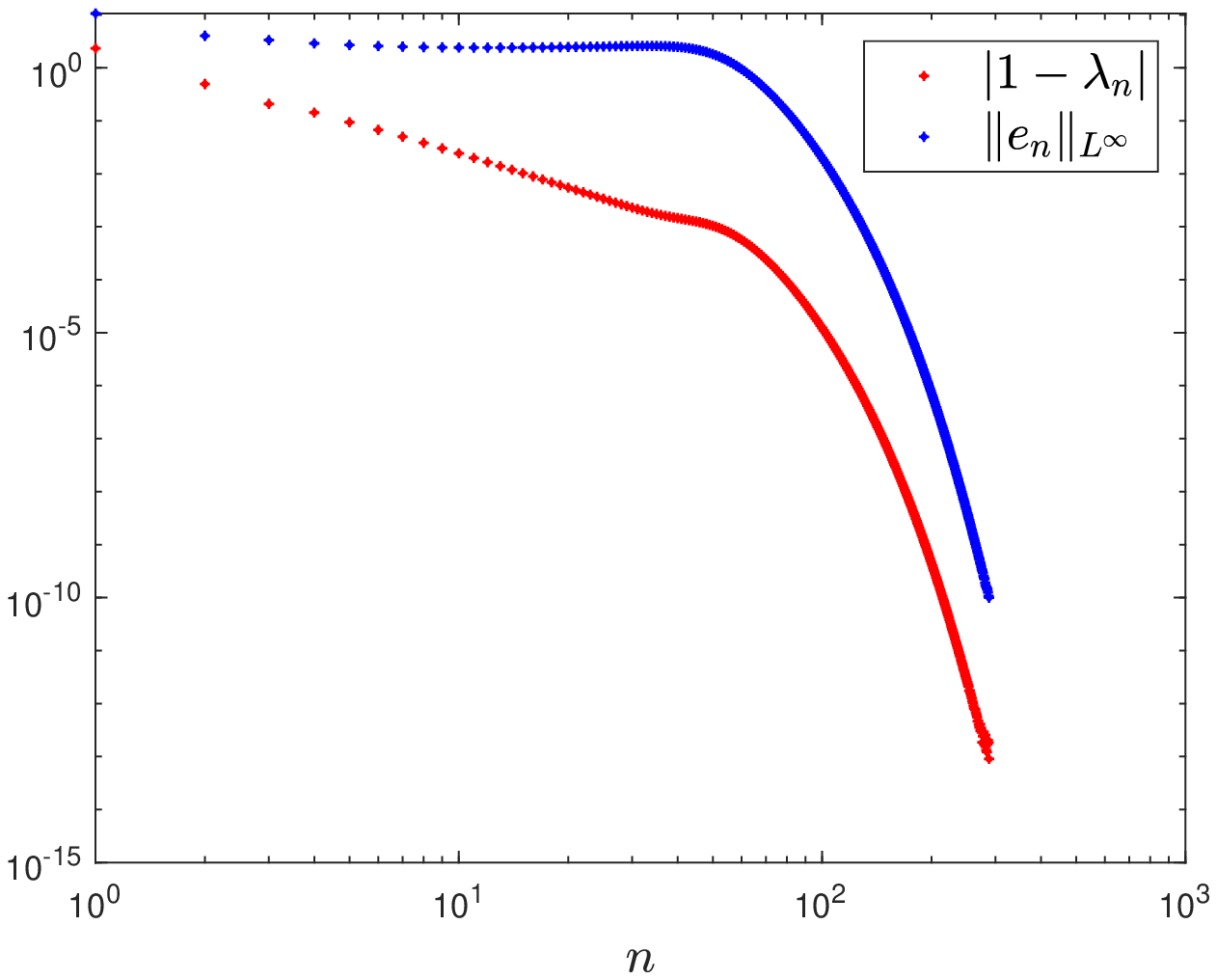}
	\includegraphics[scale=0.45]{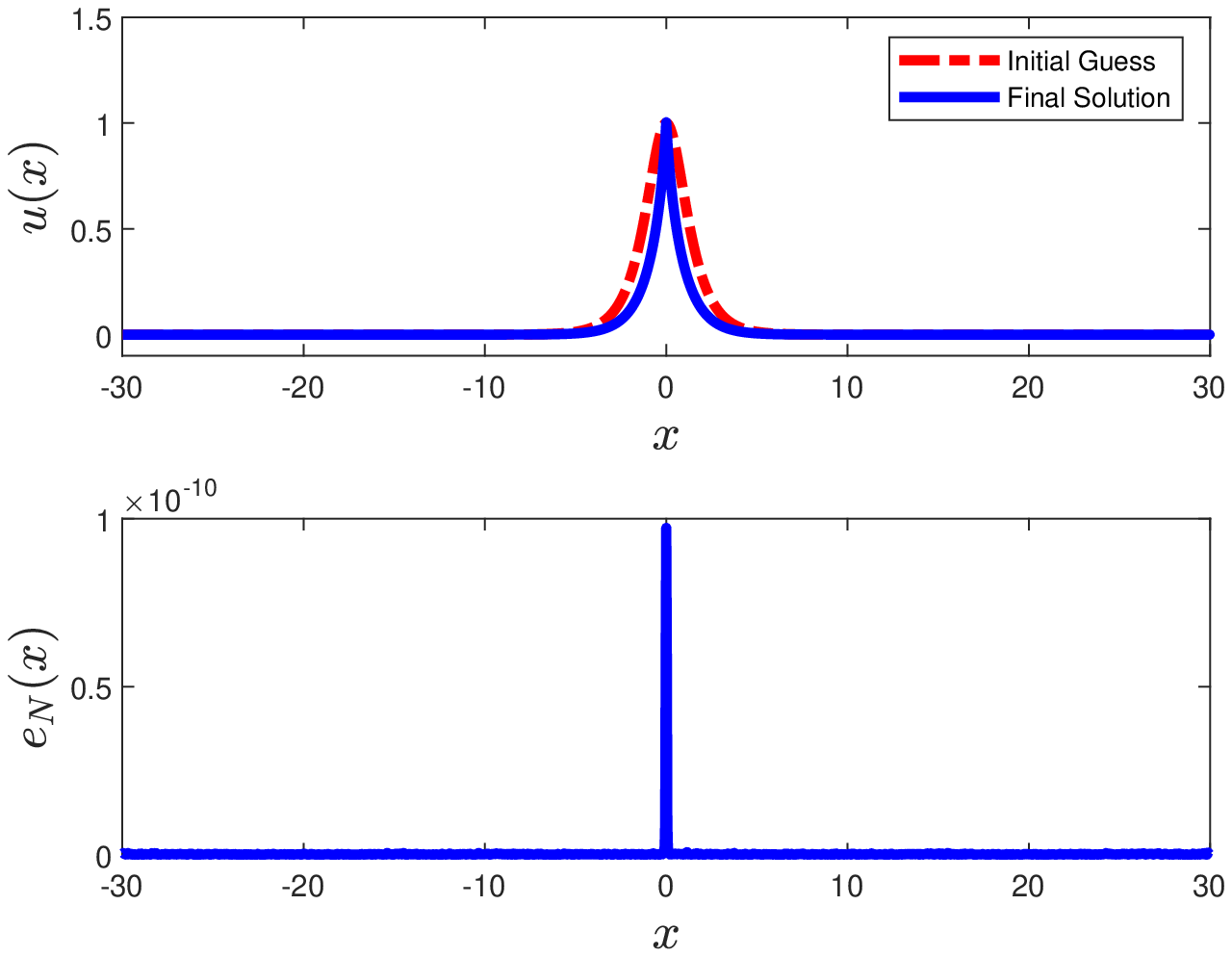}
	\caption{Left: Convergence of $|1-\lm_n|$ and $\| e_n
		\|_{L^{\infty}}$ (defined
		in the text) vs. the iteration index $n$. The former quantity converges faster than the latter one. Right: The cusped soliton obtained by the iterative method (\ref{seq-4})--(\ref{seq-5}) after $N$ iterations and the approximation error $e_N$ vs $x$.}
	\label{fig-cusped}
\end{figure}

\begin{figure}[htb]
	\includegraphics[scale=0.5]{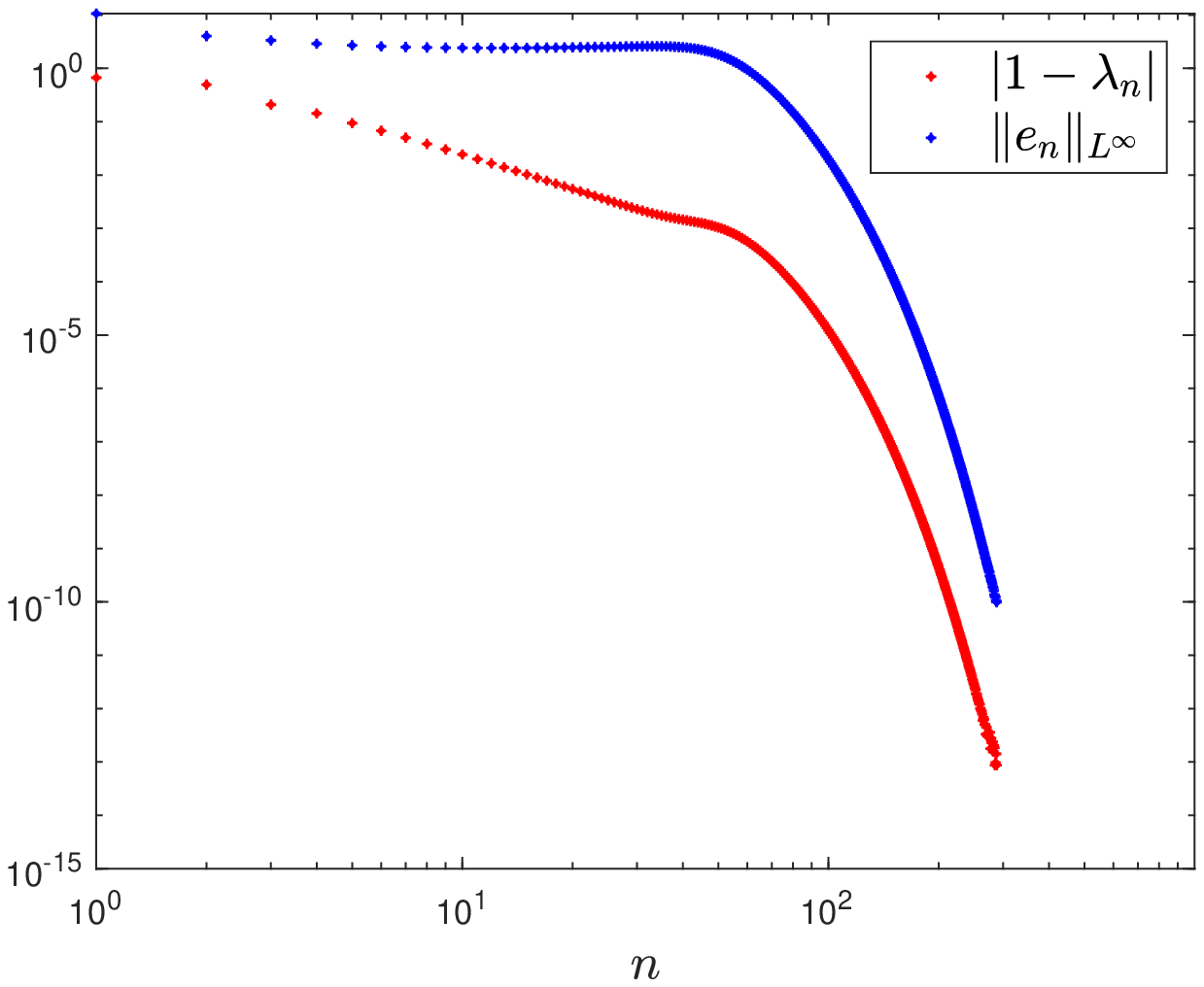}
	\includegraphics[scale=0.5]{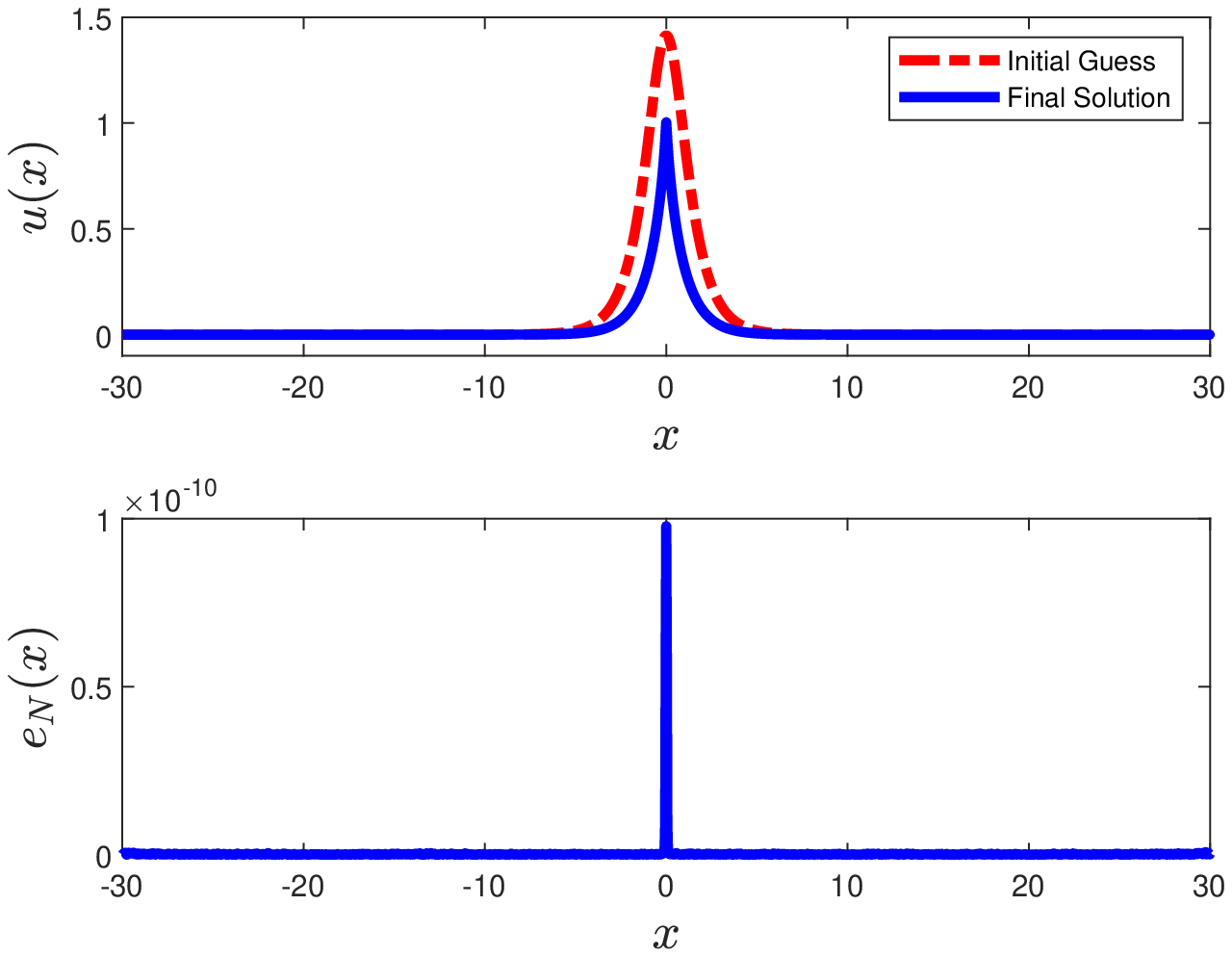}
	\caption{The same as Fig. \ref{fig-cusped} but for the initial
		guess $w_0(x)=\sqrt{2}\, { \rm \sech}(x)$.
		The Petviashvili method still converges to the cusped soliton.}
	\label{fig-sqrt2-in-p}
\end{figure}

The following proposition justifies these numerical results analytically.

\begin{proposition}
	\label{prop-Petv}
	The iterative method (\ref{seq-4})--(\ref{seq-5}) diverges 
	from the solitary wave solution $u_C$ for any fixed $C \in \mathbb{R}$.
\end{proposition}

\begin{proof}
We substitute $w_n = u + v_n \in H^1(\mathbb{R})$ and linearize the iterative method (\ref{seq-4})--(\ref{seq-5}) with respect to
$v_n$. For the normalization constant  (\ref{seq-5}), we obtain
\begin{equation}
\label{seq-6}
\lm(u+v) = 1 + \frac{2}{3 \int_{\mathbb{R}} u^2 (u')^2 dx}
\int_{\mathbb{R}} (u v + u' v_x - 3 u (u')^2 v - 3 u^2 u' v_x) dx
+ \mathcal{O}(\| v \|_{H^1}^2).
\end{equation}
Writing $\lm(u+v) = 1 + \mu(v) + \mathcal{O}(\| v \|_{H^1}^2)$
and substituting it to (\ref{seq-4}) yields the linearized iterative method:
\begin{equation}
\label{seq-7}
v_{n+1} = \frac{3}{2} \mu(v_n) u + \mathcal{M} v_n,
\end{equation}
where
\begin{equation}
\label{seq-8a}
\mu(v) := \frac{2}{3 \int_{\mathbb{R}} u^2 (u')^2 dx} \int_{\mathbb{R}} (u v + u' v_x - 3 u (u')^2 v - 3 u^2 u' v_x) dx
\end{equation}
and
\begin{equation}
\label{seq-8}
\mathcal{M} := -(1 - \partial_x^2)^{-1} (2 u u''  + u^2  \partial_x^2)
\end{equation}
is the linearized operator. Since $\mathcal{M}u = 3u$, $\mathcal{M}$ is not a contraction in $H^1(\mathbb{R})$; however it may become a contraction after a constraint imposed by $\lm$.
In order to add the constraint, we introduce the decomposition
$v_n = \alpha_n u + \xi_n$,
where $\alpha_n$ is uniquely defined under the orthogonality condition
\begin{equation}
\label{seq-9}
\int_{\mathbb{R}} (u \xi_n + u' \partial_x \xi_n - 3 u (u')^2 \xi_n - 3 u^2 u' \partial_x \xi_n) dx = 0,
\end{equation}
which yields $\mu(\alpha_n u + \xi_n) = -2 \alpha_n$.
Hence, $\alpha_{n+1} = 0$ and $\xi_{n+1} = \mathcal{M}\xi_n$.

The linearized operator (\ref{seq-7}) is a contraction in $H^1(\mathbb{R})$ if the spectrum of $\mathcal{M}$ in $H^1(\mathbb{R})$ belongs to the unit disk except for the simple eigenvalue $3$, for which the eigenvector $u$ is removed by the orthogonality condition (\ref{seq-9}). By Proposition 
\ref{prop-Petv-cusped} below, the continuous spectrum of $\mathcal{M}$ 
is given by ${\rm image}(u^2)$. Since ${\rm image}(u^2) = [0,1+e^{2C}]$ for the solitary wave solution $u_C$ with $1 + e^{2C} > 1$, $\mathcal{M}$ is not contractive for any fixed $C \in \mathbb{R}$. Hence, the iterative method (\ref{seq-4})--(\ref{seq-5}) diverges from $u_C$ 
due to the continuous spectrum of $\mathcal{M}$.
\end{proof}

\begin{proposition}
	\label{prop-Petv-cusped}
	Let $u$ be the solitary wave solution $u_C$. The continuous spectrum of the
	linear operator $\mathcal{M}$ in $H^1(\mathbb{R})$ is given by ${\rm image}(u^2) = [0,1+e^{2C}]$.
\end{proposition}

\begin{proof}
	Let us consider the spectral problem
	\begin{equation}
	\label{SL-M} \mathcal{M} v = \mu v, \quad v \in H^1(\R).
	\end{equation}	
	For simplicity, we work with the cusped soliton $u_{\rm cusp}$, 
	for which the singularity is placed at one point, $x = 0$.
	By Theorem 4 in \cite[p. 1438]{Dunford}, the continuous spectrum of $\mathcal{M}$ is given by the union of the continuous spectra of $\mathcal{M}_+$ and $\mathcal{M}_-$ restricted on $\mathbb{R}_{\pm}$ subject to the Dirichlet condition at $x = 0$:
	\begin{equation}
	\mathcal{M}_{\pm} := (1-\partial_x^2)^{-1} \left[ - u^2\partial_x^2 - \frac{2 u^2}{1 - u^2} \right] : H^1_0(\mathbb{R}_{\pm}) \mapsto H^1_0(\mathbb{R}_{\pm}),
	\end{equation}
	where $H^1_0(\mathbb{R}_{\pm})$ denote the Sobolev space of $H^1(\mathbb{R}_{\pm})$ functions vanishing at $x = 0$.
	By the symmetry of $u$, the spectrum of $\mathcal{M}_+$ is identical to the spectrum of $\mathcal{M}_-$. Hence, we consider the spectral problem for the operator $\mathcal{M}_+$ only. 
	
	The Green's function $G(x,y)$ for $(1 - \partial^2_x) G(x,y) = \delta(y)$ in $H^1_0(\mathbb{R}_+)$ is given by
	\begin{equation}
	\label{Green-expr}
	G(x,y) = \sinh(x) e^{-y} U(y-x), \quad x, y \in \mathbb{R}_+,
	\end{equation}
	where $U$ is the unit step function. Using elementary operations,
	we can rewrite the spectral problem $\mathcal{M}_+ v = \lambda v$
	for $v \in H^1_0(\mathbb{R}_+)$ in the following integral form:
	\begin{equation}
	\label{Green-function}
	v(x) = \int_x^{\infty} G(x,y)
	\frac{u(y)^2 [3- u(y)^2]}{[u(y)^2 - \lambda] [1 - u(y)^2]} v(y) dy,
	\quad x \in \mathbb{R}_+.
	\end{equation}
	The proof that $\sigma_c(\mathcal{M}_+) = {\rm image}(u^2) = [0,1]$ is
	standard. Indeed, if $\lambda \in {\rm image}(u^2)$, then the integral
	in the right-hand-side of (\ref{Green-function}) diverges
	unless $v(x_{\lambda}) = 0$ at $x_{\lambda}$ given by $u^2(x_{\lambda}) = \lambda$,
	hence the resolvent operator $(\mathcal{M}_+ - \lambda I)^{-1}$ is unbounded in $H^1_0(\mathbb{R}_+)$. Thus, 
	$\sigma_c(\mathcal{M}) = \sigma_c(\mathcal{M}_+) \cup \sigma_c(\mathcal{M}_-) = {\rm image}(u^2) = [0,1]$.
	
	For the case of the solitary wave solution $u_C$ with fixed $C \in \mathbb{R}$, the singularities are placed at two points $x = \pm \ell_C$.
	Partitioning $\mathbb{R}$ into $(-\infty,-\ell_C] \cup [-\ell_C,\ell_C] \cup [\ell_C,\infty)$ subject to the Dirichlet boundary conditions at $x = \pm \ell_C$ gives the same result $\sigma_c(\mathcal{M}) = {\rm image}(u^2) = [0,1+e^{2C}]$ but the Green's functions in (\ref{Green-function}) are expressed differently from the form (\ref{Green-expr}).
\end{proof}

\begin{remark}
	\label{remark-Petv}
Since ${\rm image}(u^2) = [0,1]$ for the cusped soliton $u_{\rm cusp}$, 
$\mathcal{M}$ is not a strict contraction for the cusped soliton $u_{\rm cusp}$.
Our numerical results on Fig. \ref{fig-M-spectrum} suggest that 
$\sigma(M) \backslash \{3\} \subset [-1,1]$ in $H^1(\mathbb{R})$, hence $\mathcal{M}$ is a contraction for the cusped soliton $u_{\rm cusp}$  under the orthogonality condition (\ref{seq-9}). Despite the lack of strict contraction, the iterative method (\ref{seq-4})--(\ref{seq-5})
converges to the cusped soliton due to discretization and truncation, 
in agreement with the numerical results shown in Fig.  \ref{fig-cusped} 
and \ref{fig-sqrt2-in-p}. 
\end{remark}

Fig. \ref{fig-M-spectrum} shows the numerical approximation of the
spectrum of $\mathcal{M}$ in $H^1(\mathbb{R})$ at the cusped soliton
$u_{\rm cusp}$ (left) and the bell-shaped soliton $u_{\rm bell}$
(right).  The numerical approximations are obtained with the Fourier
spectral method. The location of the spectrum of $\mathcal{M}$ agrees with Proposition \ref{prop-Petv-cusped} and Remark \ref{remark-Petv}.

\begin{figure}[htb]
	\includegraphics[scale=0.5]{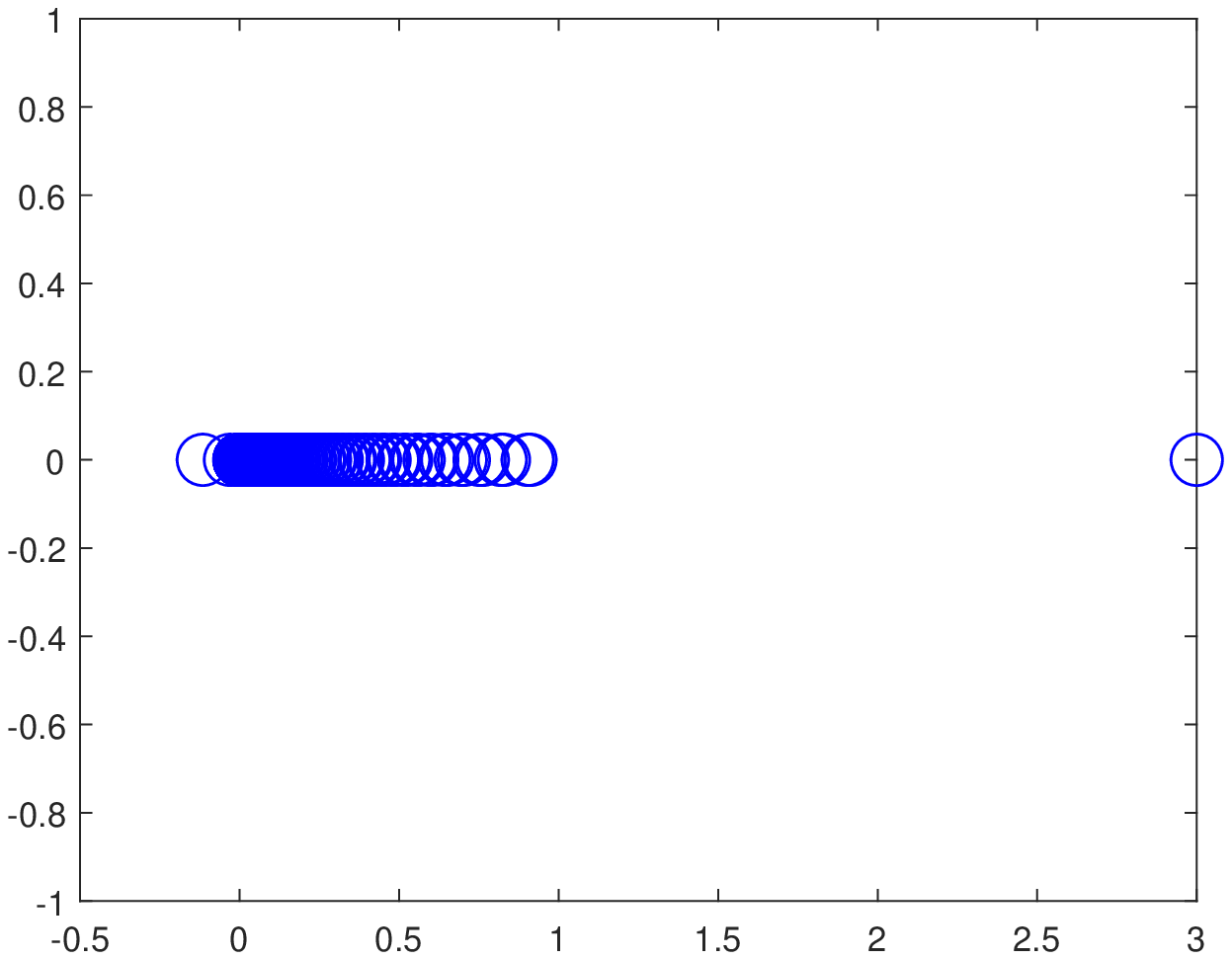}
	\includegraphics[scale=0.5]{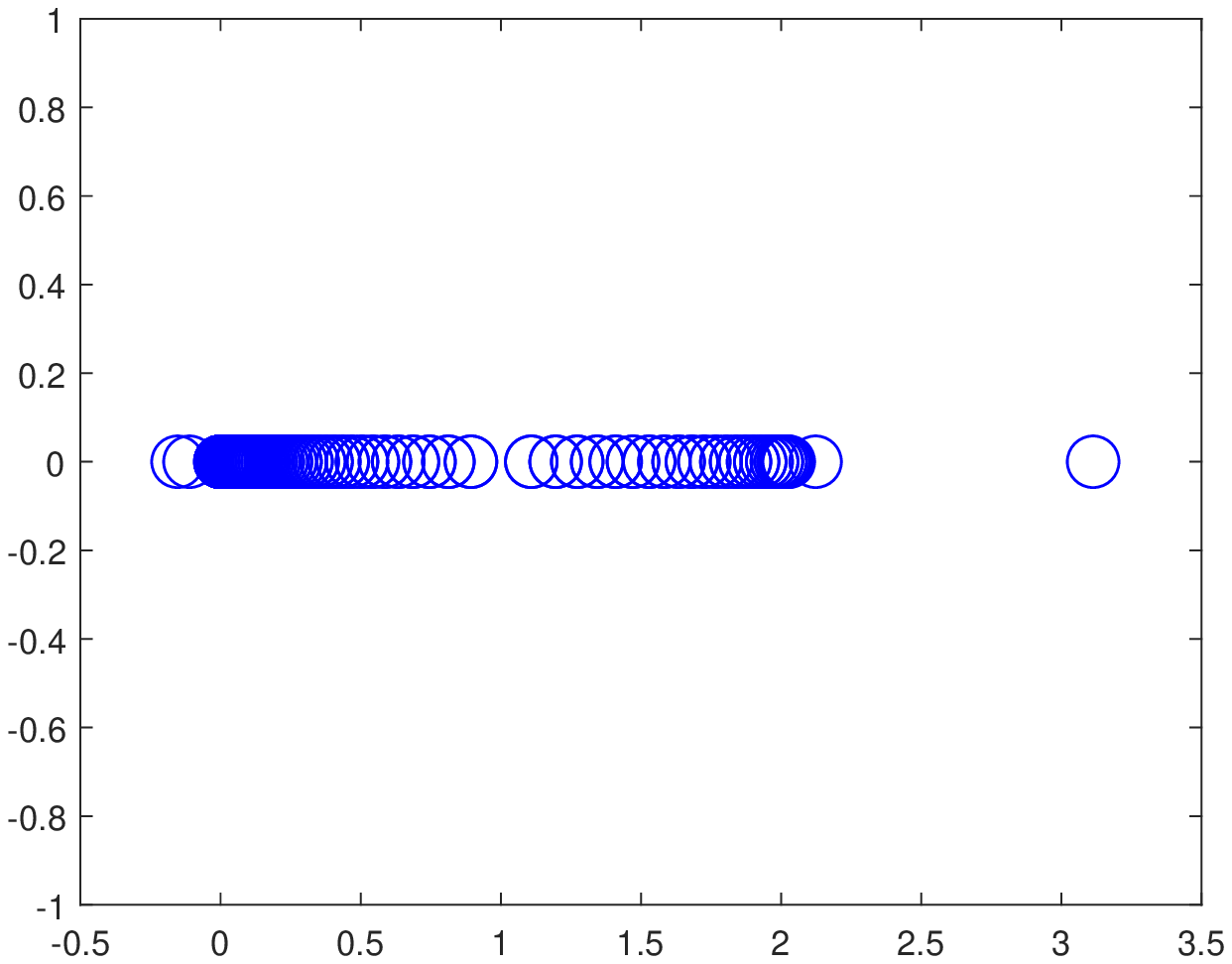}
	\caption{Numerical approximation of the spectrum of $\mathcal{M}$ in $H^1(\mathbb{R})$ at the cusped soliton (left) and the bell-shaped soliton (right).}
	\label{fig-M-spectrum}
\end{figure}

\begin{remark}
	The cusped soliton $u_{\rm cusp}$ represents the lowest energy state in the continuous family of the solitary wave solutions $u_C$. Indeed, it follows from (\ref{cont-family}) with $\ell_C> 0$ for every fixed $C \in \mathbb{R}$ that $E(u_{\rm cusp}) \leq E(u_C)$. Petviashvili's method is hence useful to approximate numerically the lowest energy state of the NLS equation.
\end{remark}

\subsection{Bell-shaped and cusped solitons via Newton method}

We represent solutions of the second-order equation (\ref{2nd}) as the roots of the nonlinear equation $F(u) = 0$, where
\begin{align}
\label{eq-F}
F(u) := -(1-u^2)\partial_x^2 u + u
\end{align}
If $u$ is a root of $F$, then performing a linearization of $F(u+v)$ with respect to $v$ leads to the linearized operator
\begin{align}
\label{eq-L}
F(u+v) = \mathcal{L} v + \mathcal O(\|v\|_{H^1}^2), \qquad \mathcal{L} := -(1-u^2)\partial_x^2 + \frac{1+u^2}{1-u^2}.
\end{align}
Roots of the nonlinear equation $F(u) = 0$ in $H^1(\R)$ can be approximated by using the Newton iterations:
\begin{align}
\label{eq-euler}
u_{n+1} = u_n - \mathcal{L}^{-1}F(u_n), \quad n \in \mathbb{N},
\end{align}
starting with any $u_1 \in H^1(\R)$ provided that $\mathcal{L}^{-1} : H^{-1}(\R) \mapsto H^1(\R)$ exists. Note the correspondence $\mathcal{M} = 1  -(1-\partial_x^2)^{-1} \mathcal{L}$ between linearized operators 
of the two methods.

The following proposition shows that $\mathcal{L}$ is invertible for the cusped soliton $u_{\rm cusp}$.

\begin{proposition}
	\label{prop-Newton-cusped}
	Let $u$ be the cusped soliton $u_{\rm cusp}$. Then, $\sigma(\mathcal{L}) = [1,\infty)$ in $L^2(\mathbb{R})$.
\end{proposition}

\begin{proof}
	Let us consider the spectral problem
	\begin{equation}
	\label{spectrum-Newton} \mathcal{L} v = \lambda v, \quad v \in {\rm Dom}(\mathcal{L}) \subset L^2(\mathbb{R}).
	\end{equation}
	Because $(1-u^2)^{-1}$ is not integrable near $x = 0$ due to singular behavior (\ref{asympt-sol}), we have $v(0) = 0$ if $v \in {\rm Dom}(\mathcal{L})$.
	Multiplying (\ref{spectrum-Newton}) by $v$ and integrating by parts under the condition $v(0) = 0$ yields 
	\begin{eqnarray*}
\lambda \int_{\mathbb{R}} v^2 dx &=& 
-\int_{\mathbb{R}} (1-u^2) v v'' dx + \int_{\mathbb{R}} \frac{1+u^2}{1-u^2} v^2 dx \\ 
&=& 
\int_{\mathbb{R}} (1-u^2) (v')^2 dx +
\int_{\mathbb{R}} (uu')' v^2 dx + \int_{\mathbb{R}} \frac{1+u^2}{1-u^2} v^2 dx \\ 
&=& 
\int_{\mathbb{R}} (1-u^2) (v')^2 dx + 
\int_{\mathbb{R}} \left[ \frac{1+2u^2}{1-u^2} + (u')^2 \right] v^2 dx. 
	\end{eqnarray*}
Since $u \in [0,1]$, we have $\lambda \geq 1$. 	Since $u(x) \to 0$ as $|x| \to \infty$ exponentially fast, Weyl's theory implies that $[1,\infty) \subset
\sigma_c(\mathcal{L})$. Hence, $\sigma(\mathcal{L}) = [1,\infty)$ for the cusped soliton $u_{\rm cusp}$.
\end{proof} 


\begin{remark}
	For the solitary wave solution $u_C$ of Proposition \ref{prop-soliton-2} with fixed $C \in \mathbb{R}$, the weight $(1-u^2)$ is no longer positive. As a result, $\sigma(\mathcal{L})$ is expected to be sign-indefinite, in agreement with the numerical approximation on Fig. \ref{fig-spectrumL} (right).
\end{remark}
 
We study invertibility of $\mathcal{L}$ numerically by rewriting the spectral 
problem $\mathcal{L} v = \lambda v$ as the generalized 
eigenvalue problem  
\begin{align}
\label{genL}
\mathcal{A} v = \lambda \mathcal{B} v, \qquad 
\mathcal{A} := -(1-u^2)^2\partial_x^2 + (1+u^2), \quad 
\mathcal{B} := (1-u^2).
\end{align}
In the form (\ref{genL}), the singularity of the potential $\frac{(1+u^2)}{(1-u^2)}$ in $\mathcal{L}$ is avoided. 
The generalized eigenvalue problem (\ref{genL}) can be solved 
numerically even when $\mathcal{B}$ is not invertible. 

\begin{figure}[ht]
	\includegraphics[scale=0.4]{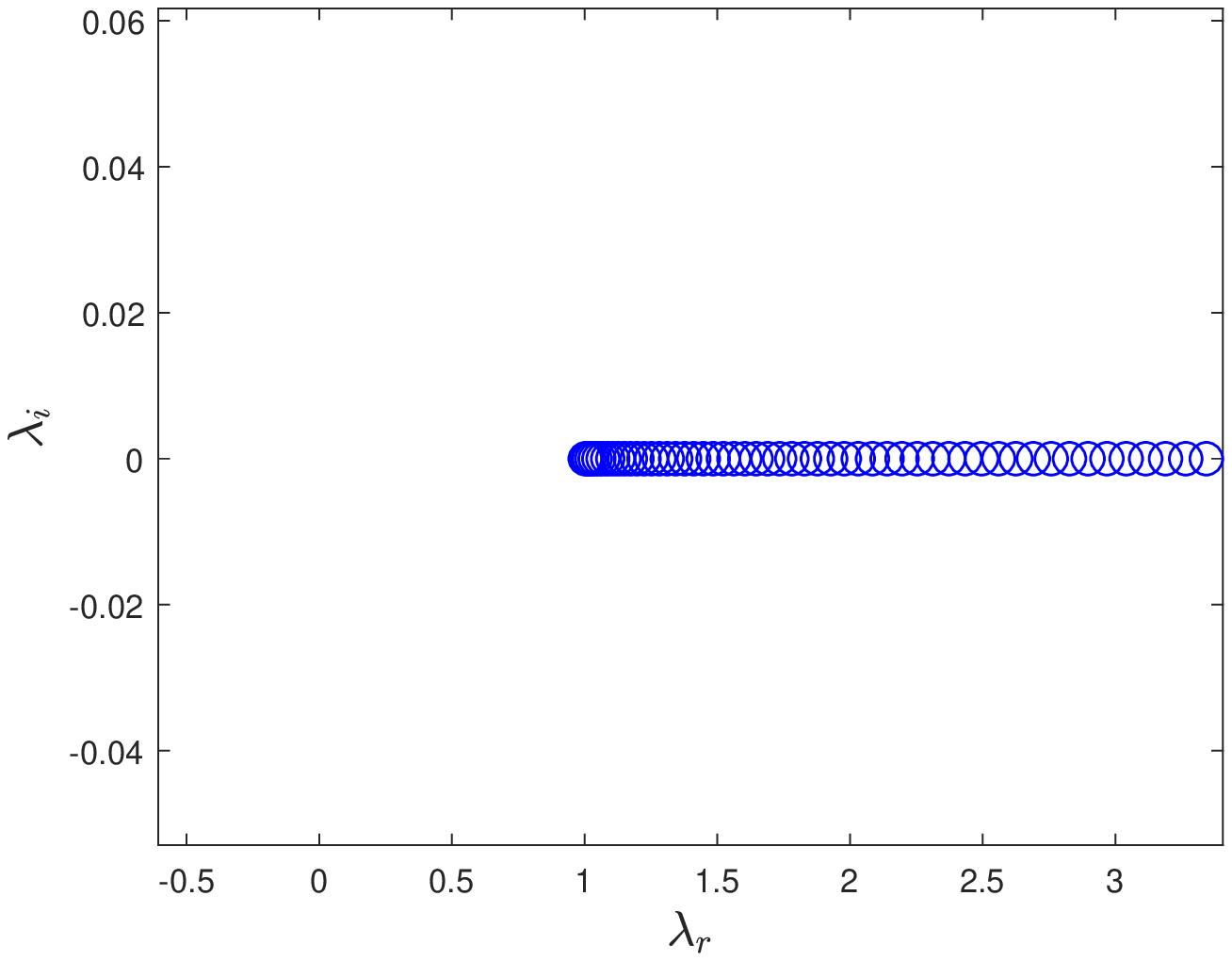}
	\includegraphics[scale=0.4]{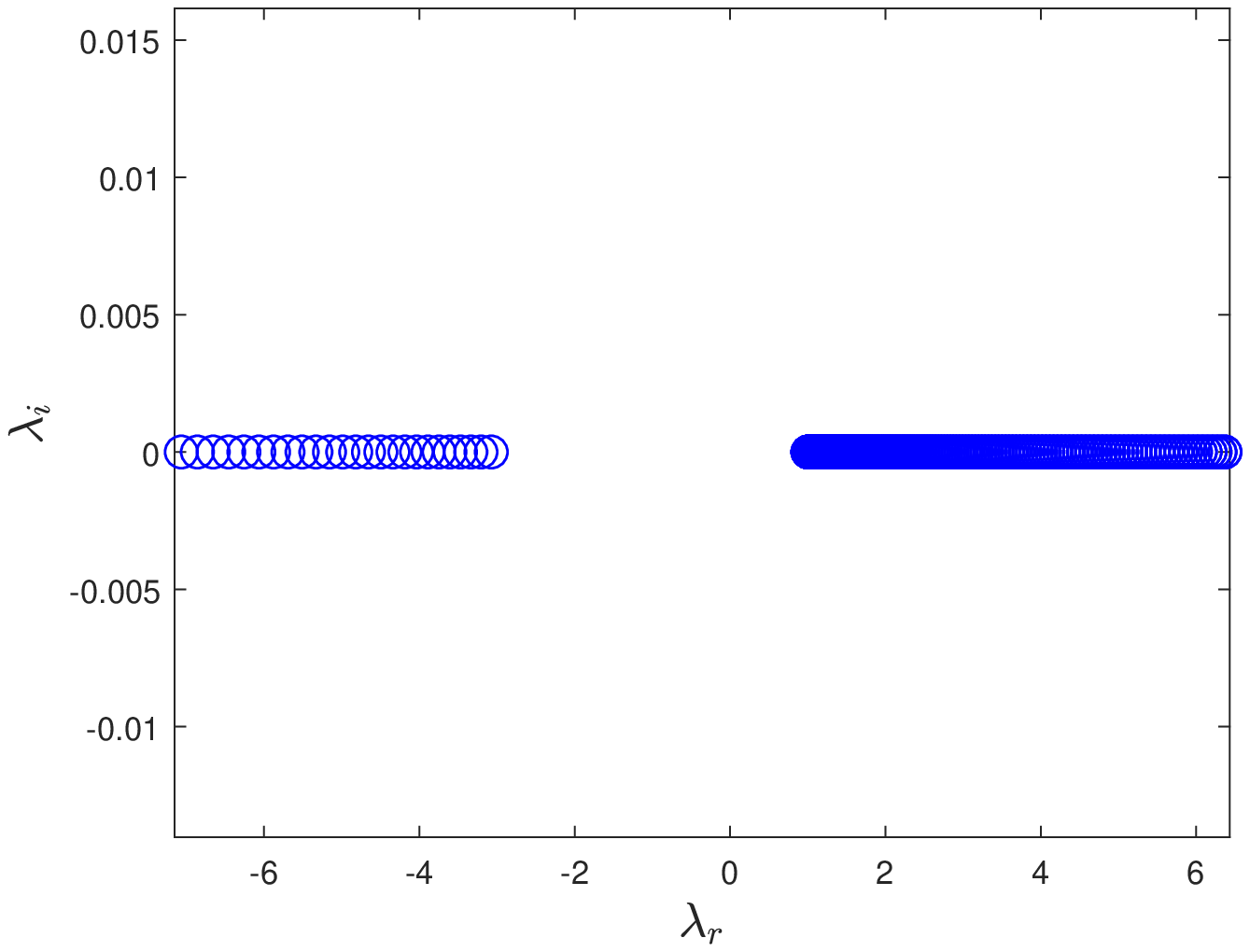}
	\caption{The spectrum of $\mathcal{L}$ approximated numerically at the cusped soliton (left) and bell-shaped soliton (right).}
	\label{fig-spectrumL}
\end{figure}

Numerical approximations of the spectrum of $\mathcal{L}$ at the cusped and bell-shaped solitons are shown in Fig. \ref{fig-spectrumL}. 
The spatial derivatives are replaced by the same Fourier 
differentiation matrices as before.  
The numerical results suggest that $\sigma(\mathcal{L}) = [1,\infty)$ for the cusped soliton $u_{\rm cusp}$ in agreement with Proposition \ref{prop-Newton-cusped} 
and $\sigma(\mathcal{L}) = (-\infty,\lambda_0] \cup [1,\infty)$ 
with $\lambda_0 < 0$ for the bell-shaped soliton $u_{\rm bell}$. In both cases, 
$\mathcal{L}$ is invertible and the Newton iterative method (\ref{eq-euler}) can be used unconditionally.

For solutions from the continuous family $u_C$, numerical results indicate that $\sigma(\mL_+) = (-\infty,\lm_C] \cup [1,\infty)$, where $\lm_C < 0$ depends on $C$. Figure \ref{fig-lambdaC} shows the location of $\lm_C$ with varying $C$. In agreement with Proposition \ref{prop-Newton-cusped}, the results suggest that $\lm_C \to -\infty$ as $C\to -\infty$, so that the spectrum of $\mathcal{L}$ at for the cusped soliton $u_{\rm cusp}$ reduces to just $[1,\infty)$.

\begin{figure}[hbt]
	\includegraphics[scale=0.4]{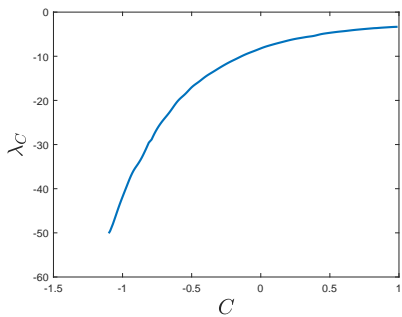}
	\caption{The dependence of $\lm_C$ on $C$ in the spectrum of $\mathcal{L}$.}
	\label{fig-lambdaC}
\end{figure}

Figures \ref{newton-cusped} and \ref{newton-bell} show the convergence
of the Newton method (\ref{eq-euler}) to the cusped and bell-shaped
soliton respectively.
As discussed earlier, a key advantage of this method is its enabling
to converge from suitable (distinct) initial guesses to both the
principal
solutions of the IDD model. Another advantage of the method
is its ability to converge to arbitrary members $u_C$ of the continuous
family
of solutions.
At each step of the iteration, we measure convergence using the
distance between successive iterates $\|u_{n+1}-u_n\|_{L^\infty}$, and
also the approximation error $e_n = |F(u_n)|$. We stop iterations at
step $N$ where the tolerance $\|e_N\|_{L^\infty} < 10^{-10}$ is
reached.

\begin{figure}[hbt]
	\includegraphics[scale=0.45]{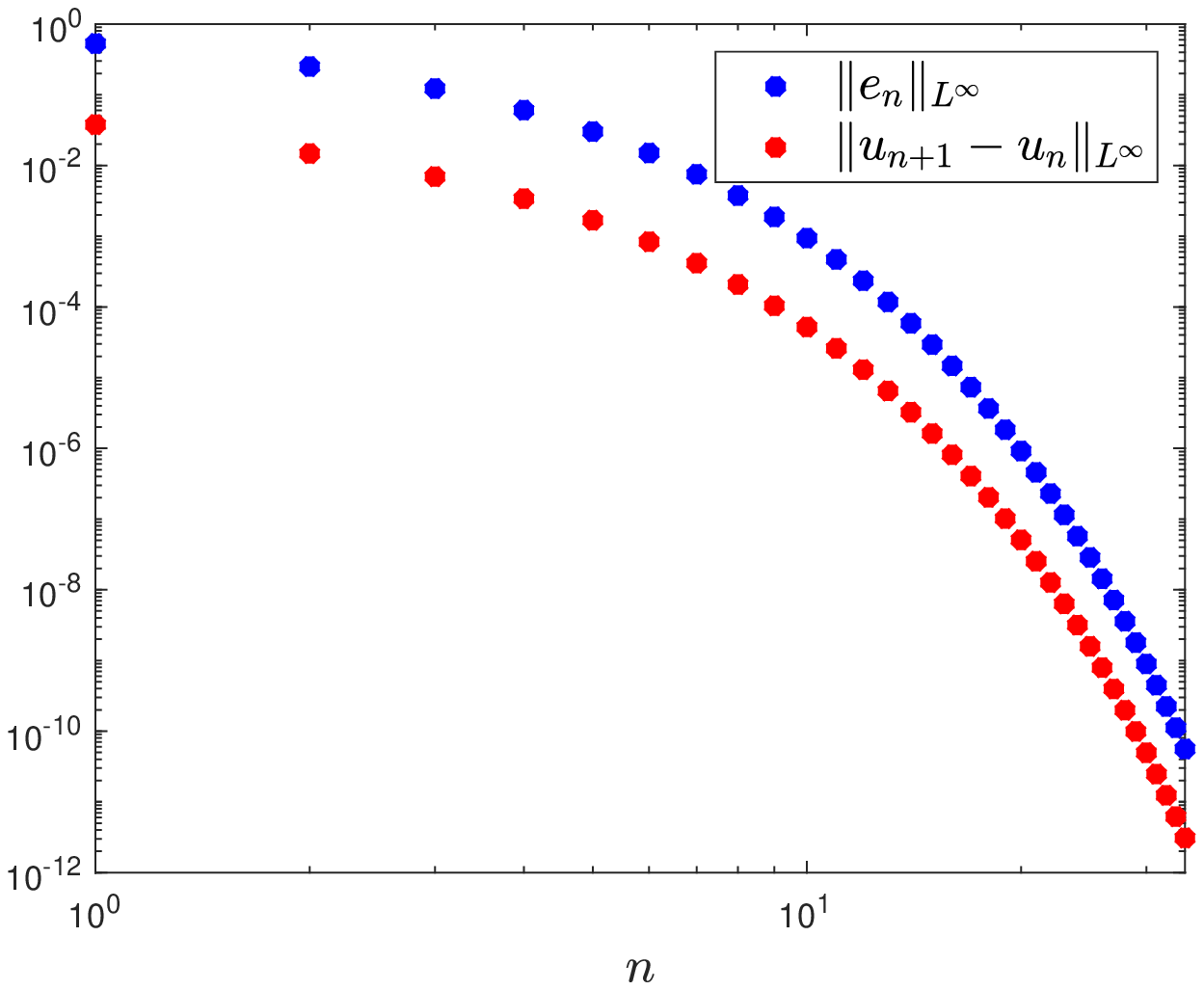}
	\includegraphics[scale=0.45]{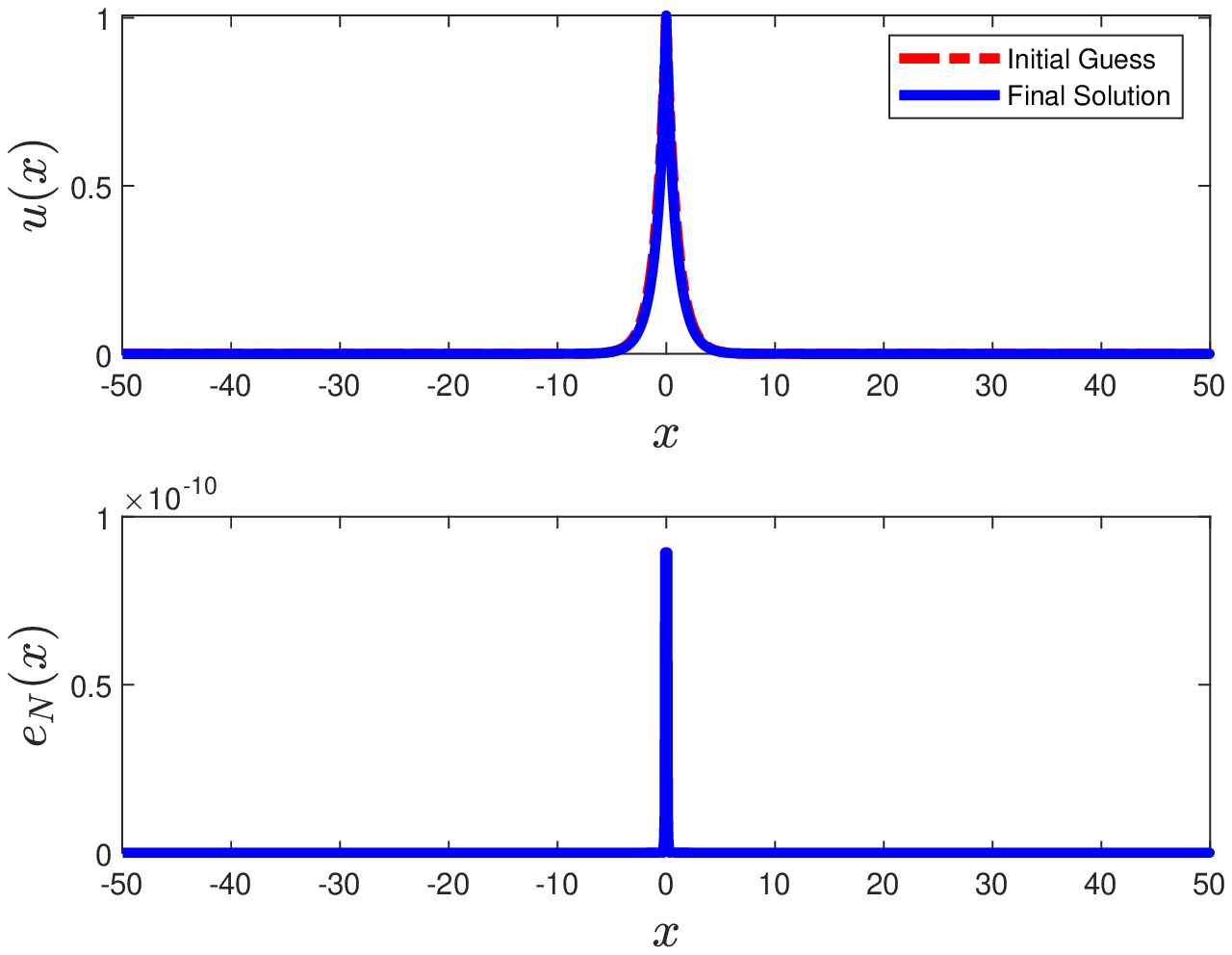}
	\caption{Based on the same diagnostics as discussed
		previously,
		the convergence of the Newton method to the cusped
		soliton is quantified, starting from the initial guess $u_1(x) =
		e^{-|x|}$.}
	\label{newton-cusped}
\end{figure}

\begin{figure}[ht]
	\includegraphics[scale=0.45]{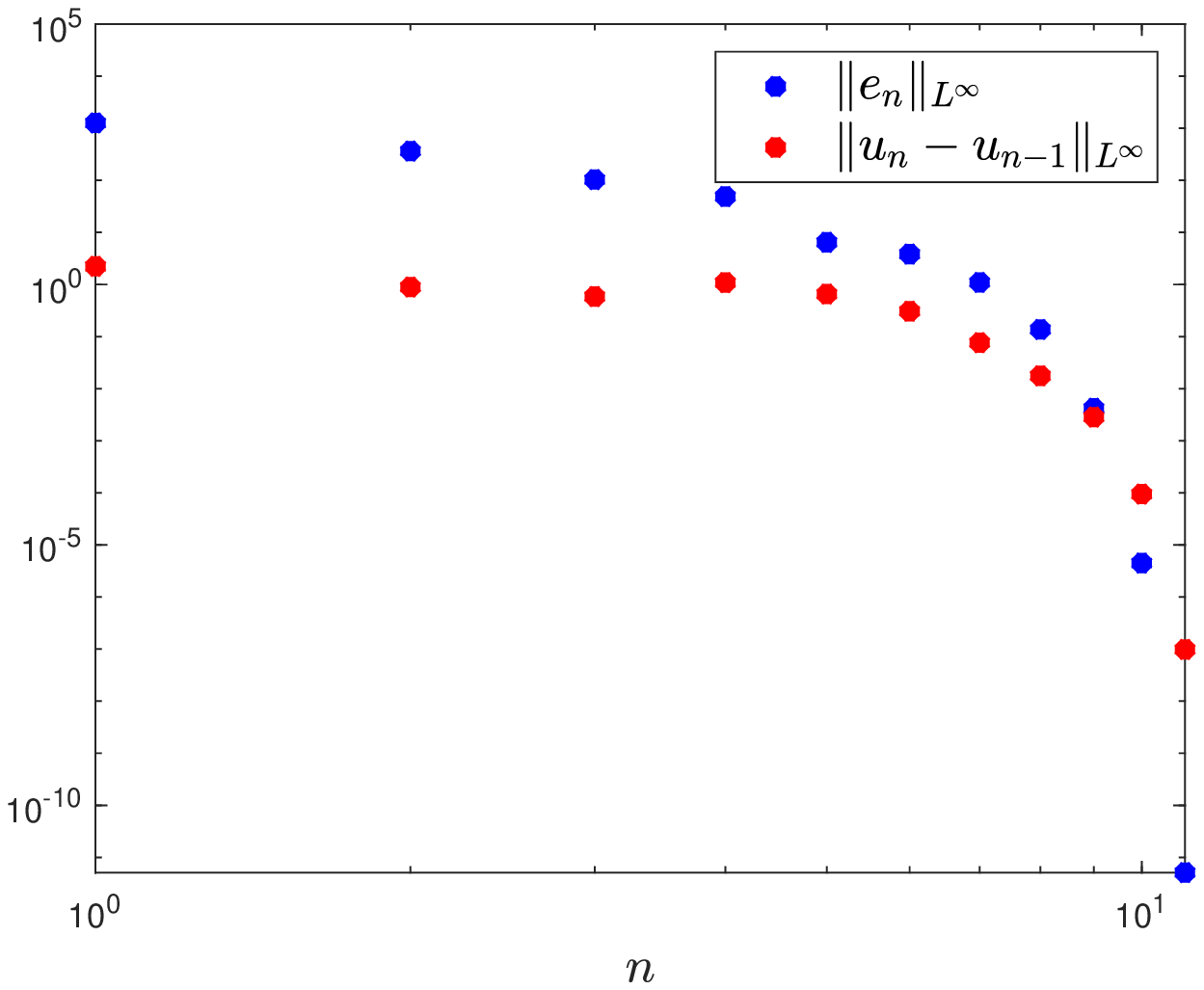}
	\includegraphics[scale=0.45]{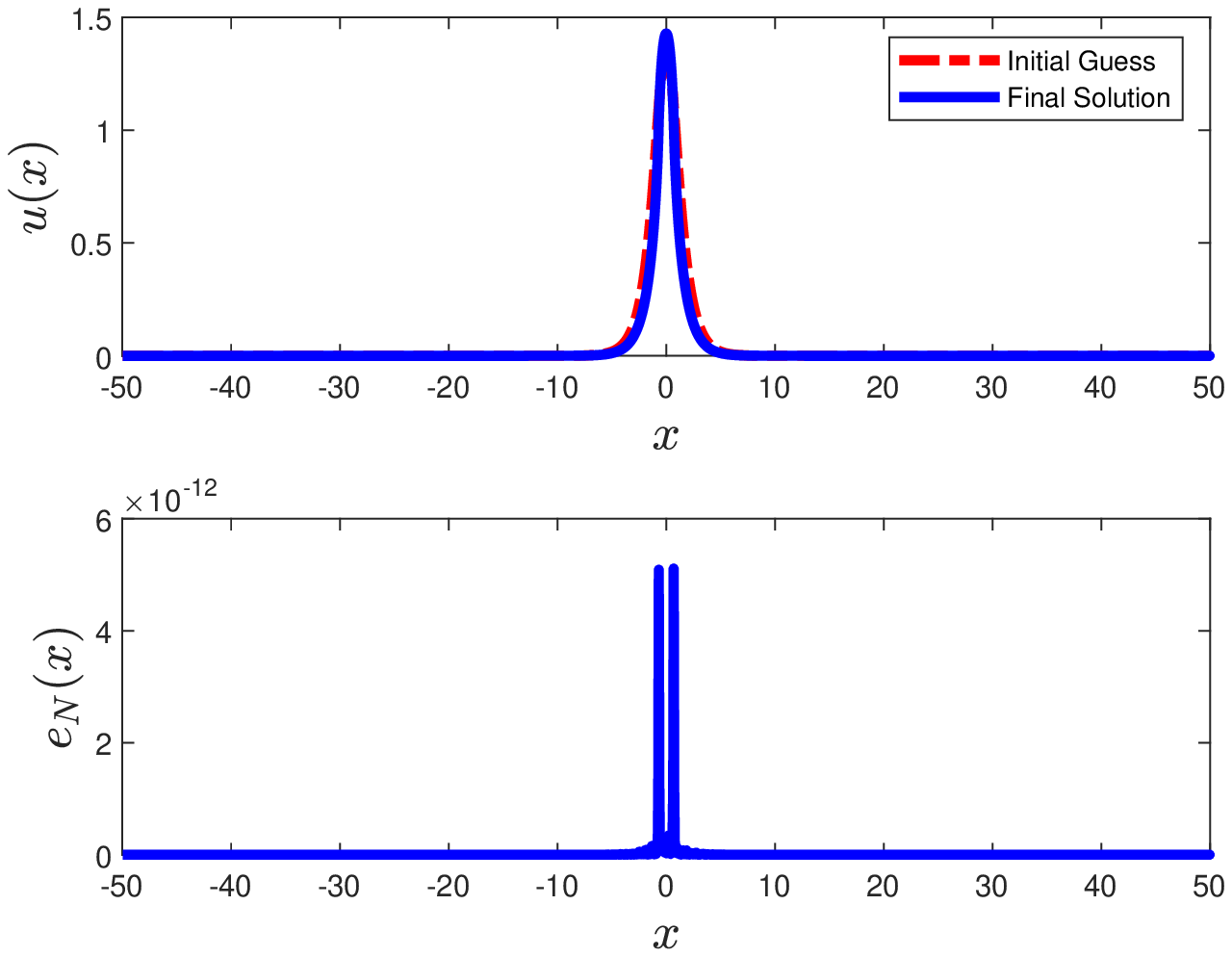}
	\caption{Convergence of the Newton method to the bell-shaped soliton starting from the initial guess $u_1(x)=\sqrt{2} {\rm sech}(x)$.}
	\label{newton-bell}
\end{figure}

\section{Spectral stability}

In order to derive the spectral stability problem, we use the ansatz
\begin{align}
\label{ansatz}
\psi(x,t) = e^{it} [u(x) + e^{\lm t}(v(x)+i w(x)) + e^{\bar{\lm} t} (\bar{v}(x) + i \bar{w}(x))]
\end{align}
where $u$ is the solitary wave solution, 
$v + iw$ is a small perturbation with real $v$ and $w$, $\lm$ is the spectral parameter,
and the bar denotes complex conjugation. Substituting (\ref{ansatz}) into the NLS equation (\ref{nls}) and linearizing in $v$ and $w$ gives the spectral stability problem
\begin{align}
\label{eq-spec-stability}
\left[ \begin{array}{cc} 0 & \mathcal{L}_- \\
-\mathcal{L}_+ & 0 \end{array} \right] \left[ \begin{array}{c} v \\ w \end{array} \right] = \lm \left[ \begin{array}{c} v \\ w \end{array} \right],
\end{align}
where 
\begin{align*}
\mathcal{L}_+ := -(1-u^2) \partial_x^2 + \frac{1+u^2}{1-u^2}, \quad 
\mathcal{L}_- := & -(1-u^2) \partial_x^2 + 1
\end{align*}
Note that $\mathcal{L}_+ \equiv \mathcal{L}$ is the linearized operator 
in the Newton method. We consider the spectral problem 
(\ref{eq-spec-stability}) in $L^2(\mathbb{R})$ so that 
the differential operators $\mathcal{L}_{\pm}$ are closed in 
$$
\mathcal{L}_{\pm} : {\rm Dom}(\mathcal{L}_{\pm}) \subset L^2(\R) \mapsto L^2(\R),
$$
where ${\rm Dom}(\mathcal{L}_{\pm}) = \{ v \in L^2(\R), \quad (1-u^2) v'' \in L^2(\R) \}$. The following proposition characterizes the zero eigenvalue of the spectral problem (\ref{eq-spec-stability}) for the cusped soliton. 

\begin{proposition}
	\label{prop-kernel}
	Let $u$ be the cusped soliton $u_{\rm cusp}$.
	The spectral problem (\ref{eq-spec-stability}) has a double zero eigenvalue
	$\lambda = 0$ in $L^2(\mathbb{R})$.
\end{proposition}

\begin{proof}
	Due to the phase rotation symmetry of the NLS equation (\ref{nls}), 
	we have 
	\begin{equation}
	\label{ker-1}
	\mathcal{L}_- u = -(1-u^2) u'' + u = 0, 
	\end{equation}
	where $u \in H^1(\mathbb{R})$ and $(1-u^2) u'' \in H^1(\mathbb{R})$. 
	Hence $(v,w) = (0,u)^T$ is the eigenvector of the spectral problem 
	(\ref{eq-spec-stability}) for $\lambda = 0$. 
	
	Due to the translation symmetry of the NLS equation (\ref{nls}), 
	we also have 
	\begin{equation}
	\label{ker-1a}
	\mathcal{L}_+ u' = -(1-u^2) u''' + \frac{1+u^2}{1-u^2} u' = 0, 
	\end{equation}
	with $u' \in L^2(\mathbb{R})$, however, $(1-u^2) u''' \notin L^2(\mathbb{R})$ 
	due to the singular behavior (\ref{asympt-sol}).
	Therefore,$(v,w) = (u',0)$ is not in the domain 
	of the spectral problem (\ref{eq-spec-stability}) and 
	the geometric multiplicity of $\lambda = 0$ is one. 
	
	It remains to check the Jordan blocks associated with the 
	eigenvector $(v,w) = (0,u)$. The first generalized eigenvector 
	satisfies the nonhomogeneous equation 
	\begin{align}
	\label{ker-2}
	\mathcal{L}_+ v = -u.
	\end{align}
	Since the kernel of $\mathcal{L}_+$ is trivial in 
	${\rm Dom}(\mathcal{L}_+) \subset L^2(\mathbb{R})$, 
	Fredholm's alternative theorem implies that there exists 
	$v \in {\rm Dom}(\mathcal{L}_+) \subset L^2(\mathbb{R})$ 
	that solves the nonhomogeneous equation (\ref{ker-2}). 
	For the cusped soliton, the solution can be found in the explicit form:
	\begin{equation}
	\label{v-solution}
	v = \frac{1}{2} x u',
	\end{equation}
	since $x u' \in H^1(\mathbb{R})$ and 
	$(1-u^2) (xu')'' \in L^2(\mathbb{R})$ due to the singular behavior (\ref{asympt-sol}). 
	
	The second generalized eigenvector, if it exists, satisfies 
	the nonhomogeneous equation 
	\begin{align}
	\label{ker-3}
	\mathcal{L}_- w = v,
	\end{align}
	where $v$ is the solution to the nonhomogeneous equation 
	(\ref{ker-2}). However, no solution $w \in {\rm Dom}(\mathcal{L}_-) \subset L^2(\mathbb{R})$ exists by the Fredholm alternative if
	\begin{equation}
	\label{ker-3a}
	\langle \frac{v}{1-u^2}, u \rangle \neq 0,
	\end{equation}
	where we have used the fact that $(1-u^2)^{-1}\mathcal{L}_-$ is a symmetric differential operator.
	Using (\ref{v-solution}) and integrating by parts, we obtain 
	\begin{equation}
	\label{ker-3b}
	\frac{1}{2} \langle \frac{x u'}{1-u^2}, u \rangle = 
	-\frac{1}{4} \int_{\mathbb{R}} \log(1-u^2) dx \neq 0,
	\end{equation}
	where the integration by parts is justified since 
	$x \log(1-u^2) \to 0$ both as $|x| \to 0$ and as $|x| \to \infty$. 
	Hence, the algebraic multiplicity of $\lambda = 0$ is two.
\end{proof}

\begin{remark}
	The proof of Proposition \ref{prop-kernel} can be extended 
	to the solution $u_C$ with $C \in \mathbb{R}$ up to the nonhomogeneous equation 
	(\ref{ker-2}). However, $(1-u^2) (xu')'' \notin L^2(\mathbb{R})$ 
	due to the singular behavior (\ref{asympt-sol-bell-cont}), hence $v$ 
	is not available in the closed form.
	Consequently, it is not clear how to show that 
	the constraint (\ref{ker-3a}) holds for the solution $u_C$. 
\end{remark}

\begin{remark}
	\label{remark-kernel}
	Besides the result of Proposition \ref{prop-kernel}, it is not clear
	how to analyze the spectral problem (\ref{eq-spec-stability}) 
	and to prove spectral stability of the cusped soliton. 
	Since $u(x) \to 0$ as $|x| \to \infty$ exponentially fast, 
	the continuous spectrum of the problem 
	is located at $\lambda \in i (-\infty,-1] \cup i [1,\infty)$. 
	However, $\mathcal{L}_{\pm}$ are not symmetric 
	differential operators compared to the spectral stability problems arising in other NLS-type equations. The case 
	of the solution $u_C$ is even more difficult since $1 - u^2$ is no longer sign-definite.
\end{remark}

We approximate the spectral stability problem (\ref{eq-spec-stability}) numerically by using the Fourier discretization matrices. Figure \ref{fig-spec-stab} shows the location of the spectrum 
in the discretized and truncated system (\ref{eq-spec-stability})
for the cusped (left) and bell-shaped (right) solitons.
The double zero eigenvalue is detached from
the continuous spectrum located on  $i (-\infty,-1] \cup i [1,\infty)$ 
(see Remark \ref{remark-kernel}).
Both solitons appear to be spectrally stable for perturbations in $L^2(\R)$.

\begin{figure}[htpb]
	\includegraphics[width=8cm,height = 6cm]{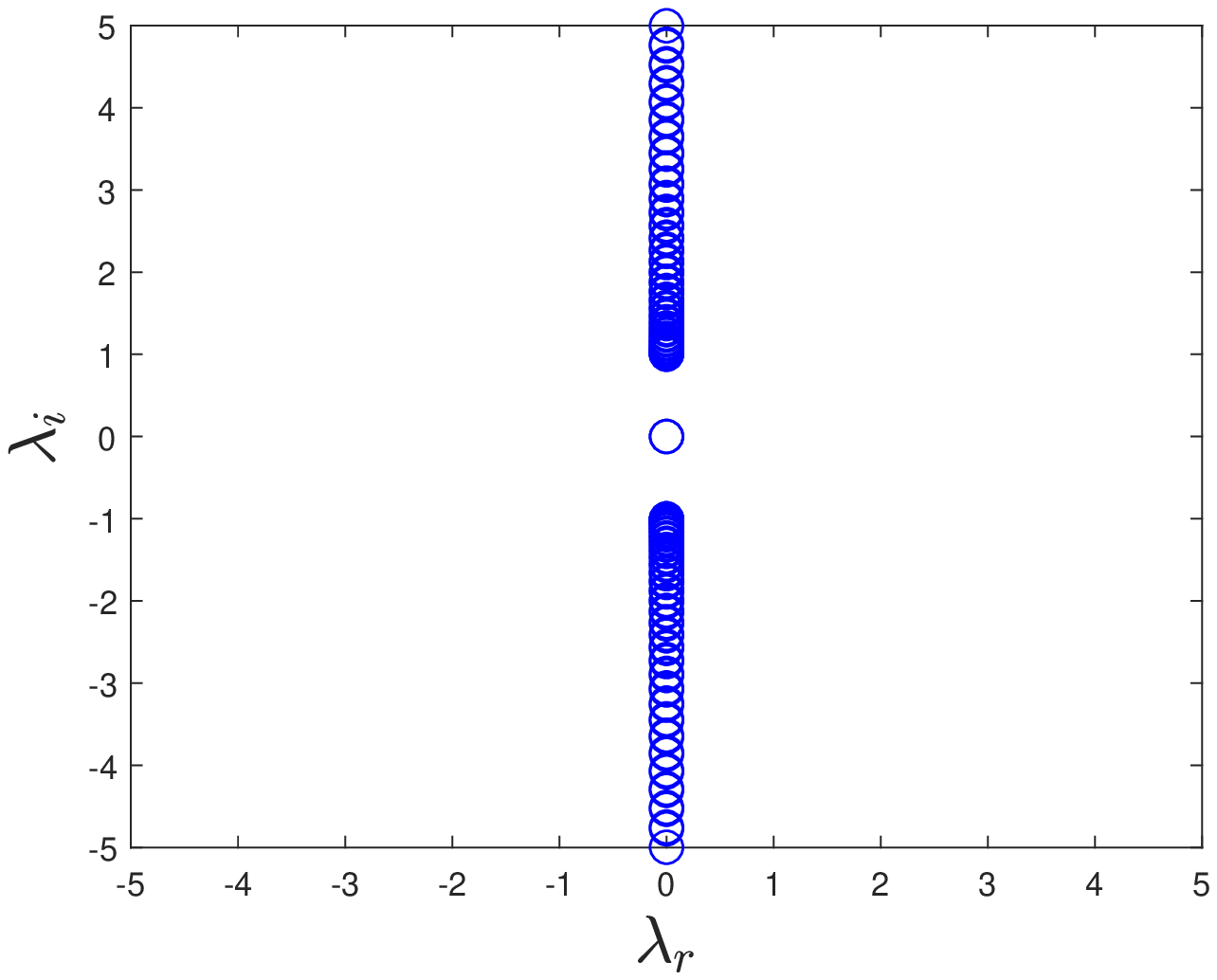}
	\includegraphics[width=8cm,height = 6cm]{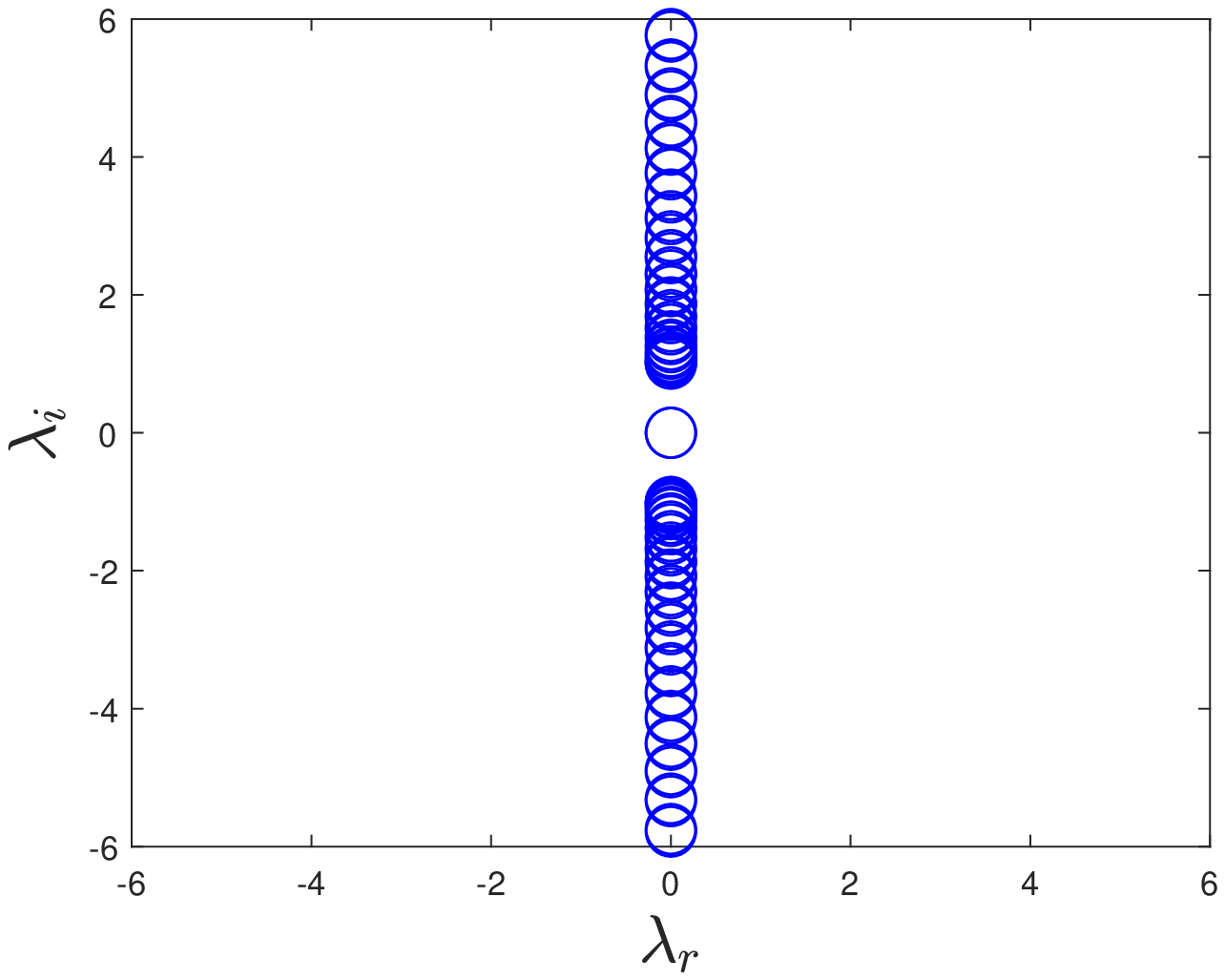}
	\caption{Eigenvalues for the spectral stability problem of the
		form of Eq.~(\ref{eq-spec-stability}) are given for the
		cusped soliton (left) and bell-shaped soliton (right).
		The spectral plane $(\lambda_r,\lambda_i)$ of the eigenvalues
		$\lambda=\lambda_r + i \lambda_i$ is shown.
		Both solutions are numerically found to be spectrally stable since there are no eigenvalues with positive real part.}
	\label{fig-spec-stab}
\end{figure}

\section{Conclusion}

We have revisited the prototypical model of the NLS equation 
with IDD. We have illustrated that the solutions of the model depend 
on the interplay between the constant coefficient dispersion 
and the intensity-dependent dispersion. We proved that no solitary
wave solutions exist when the two dispersion contributions bear
the same sign. On the other hand, the competition between
the two leads to a continuous family of solitary wave solutions 
which are singular at the points of vanishing dispersion, 
where the constant and intensity-dependent
dispersion prefactors cancel each other.

We analyzed three numerical methods which can be used to approximate
the relevant solitary waves and showed that the Newton method 
is robust in approximating different solitary waves of the
continuous family,
while the regularizing method only converges to the bell-shaped soliton 
and the Petviashvili method only converges to the cusped soliton 
being the lowest energy state of the family. 
We illustrated numerically that the bell-shaped and cusped soliton 
are spectrally stable in the time-dependent NLS equation.

A number of important outstanding questions remain for further 
development of the mathematical analysis,
including the rigorous proof of the spectral and orbital stability of the
solitary waves for which our numerical computations suggest that
they are stable.
There are also numerous extensions of this class of models that are in their infancy. First off, it is especially relevant
to consider if a cancellation of the cubic local nonlinear term (as,
e.g., is achieved in BECs in the context of the so-called Feshbach
resonances~\cite{chin}) could be engineered to realize the
prototypical
IDD model studied herein. Also, one can envision generalizations of
the IDD model depending on $|\psi|^{2k}$ and consider the nature
of the emerging solutions as a function of $k$ (as has been done
in the context of the regular NLS equation ~\cite{sulem}).
In the same spirit one can envision more complex functional forms
of the dispersion coefficient, potentially bearing multiple
zero-crossings
and attempt to classify the ensuing nonlinear waveforms.
Moreover,
one can consider lattice-based variants, which could potentially be related to waveguides in optics~\cite{moti} or
associated with optical lattices in BEC~\cite{markus}. Finally,
a natural issue to consider would be the extension of such models
into isotropic or anisotropic  generalization possibilities into higher
dimensions
and the dynamics of solitary waves that can arise therein.

\end{document}